\documentclass[10pt]{article}

\RequirePackage{amsthm,amsmath}

\usepackage{amsfonts}

\usepackage{graphicx}

\newtheorem{thm}{Theorem}[section]

\makeatletter

\renewcommand\@biblabel[1]{}

\renewenvironment{thebibliography}[1]
     {\section*{\refname}%
      \@mkboth{\MakeUppercase\refname}{\MakeUppercase\refname}%
      \list{}%
           {\leftmargin0pt
            \@openbib@code
            \usecounter{enumiv}}%
      \sloppy
      \clubpenalty4000
      \@clubpenalty \clubpenalty
      \widowpenalty4000%
      \sfcode`\.\@m}
     {\def\@noitemerr
       {\@latex@warning{Empty `thebibliography' environment}}%
      \endlist}

\makeatother

\newcommand\starfootnote[1]{%

  \begingroup

  \renewcommand\thefootnote{$*$}\footnote{#1}%

  \addtocounter{footnote}{-1}%

  \endgroup

}

\begin{document}

\title{Improving Grid Based Bayesian Methods}

\maketitle

\begin{center}

\author{Chaitanya Joshi, Paul T. Brown, and Stephen Joe\starfootnote{All three authors are with the Department of Mathematics and

Statistics, The University of Waikato, Private Bag 3105, Hamilton, New Zealand.

Their respective e-mail addresses are cjoshi@waikato.ac.nz, ptb2@students.waikato.ac.nz, and stephenj@waikato.ac.nz.}}

\end{center}

\begin{abstract}

In some cases, computational benefit can be gained by exploring the hyper parameter space using a deterministic set of grid points instead of a Markov chain. We view this as a numerical integration problem and make three unique contributions. First, we explore the space using low discrepancy point sets instead of a grid. This allows for accurate estimation of marginals of any shape at a much lower computational cost than a grid based approach and thus makes it possible to extend the computational benefit to a hyper parameter space with higher dimensionality (10 or more). Second, we propose a new, quick and easy method to estimate the marginal using a least squares polynomial and prove the conditions under which this polynomial will converge to the true marginal. Our results are valid for a wide range of point sets including grids, random points and low discrepancy points. Third, we show that further accuracy and efficiency can be gained by taking into consideration the functional decomposition of the integrand and illustrate how this can be done using anchored f-ANOVA on weighted spaces.

\end{abstract}


\section{Introduction}\label{intro}

For some classes of models,  it may be possible to perform a computationally efficient approximation to the posterior using methods which explore the hyper parameter space using a deterministic set of grid points. This includes models with latent Gaussian Markov random field priors, see, for example, Simpson et.\ al.\ (2016), Lindgren and Rue (2011), and Rue et.\ al.\ (2009),  but also others, for example Joshi (2011), Ormerod (2011), and Austad and Friel (2010). The support may be identified by first finding the mode of the posterior distribution and then identifying a region around the mode where the density is significantly different to zero. The support is then explored using a set of grid points centered around the mode. Subject to the computational requirements of the functional approximation, such approaches can outperform the standard Markov chain Monte Carlo (MCMC) methods in terms of the computational time when the number of hyper parameters is small. Partly due to the availability of easy to use packages such as R-INLA (Martins et.\ al.\ 2012), such approaches have been widely used across many applications. Assuming that the support has been correctly identified, the accuracy of the inference depends on firstly, the accuracy of the functional approximation to the true posterior and secondly, the set of grid points used. No formal mathematical results on the accuracy of such approximations have yet been published however.

 Two common criticisms (for example, see Yoon and Wilson (2011)) of such an approach are that it may not accurately capture the shape of the distribution and that the computational benefit over an MCMC approach may be lost if the dimensionality of the parameter space\footnote{From here on throughout the rest of the paper, we use the term \emph{parameter space} to refer to the hyper parameter space unless otherwise specified.} is even moderately large (6 or above). The first criticism could relate to both the functional approximation used and to the fact that the posterior (or an approximation to it) has been evaluated at a set of grid points in the space. But it can be shown that even when the functional approximation is exact, the marginal distributions obtained using grid based point sets can fail to capture the shape of the distribution if the distribution is highly skewed or multi-modal.  The second criticism is also (at least in part) due to the use of the grid points since the number of grid points increases exponentially with the number of dimensions.

In this paper, we view this posterior estimation as a numerical integration problem. We (a) explore the space using more optimal point sets and (b) propose a least squares method to estimate the marginal distributions.  We argue that both of the above criticisms can be overcome if the set of points used is that generated by a \emph{low discrepancy sequence} (LDS). For convenience, we may, sometimes, refer to these points simply as LDS points.  The computational advantage is very significant and real; the same level of accuracy can be achieved using a set of LDS points several times smaller than the number of grid points required as illustrated by our examples. We will also show that for smaller dimensions, where a grid can feasibly be used, using an LDS point set with similar number of points provides much higher accuracy. Later we show that even further accuracy and computational benefit can be achieved by incorporating the functional decomposition of the integrand.

For the remainder of the paper, we will assume that the functional approximation being used is exact; that is we are evaluating the true posterior distribution at a given set of points. We also assume that the support is exactly known or has been identified correctly. We will measure the accuracy of the (least squares) approximation by comparing the marginal posterior distribution obtained using a set of points to the true marginal posterior distribution. We give sufficient conditions under which this approximation will converge to the true marginal. Our results are valid for a wide range of point sets including grids, random points and LDS points.

The paper is organised as follows. In the remainder of this section, we will provide a brief introduction to  LDS and define the notation. In Section $2$ we will discuss the grid based approaches to approximate the Bayesian posterior and propose our new method. The theoretical results will be presented in Section $3$.  In Section $4$, we propose taking into consideration the functional decomposition of the integrand to further improve the computational benefit, followed by examples in Section $5$. Finally, we summarise and discuss in Section $6$.

\subsection{Low discrepancy sequences} \label{lds}

For a Riemann integrable function $f\,:\,[0,1)^{s} \rightarrow \mathbb{R},$ a typical approximation of the integral

\begin{equation} \label{lds1}
I =\int_{[0,1)^{s}} f(\boldsymbol{x}) \, d\boldsymbol{x} \hspace{1cm} \mbox{  will be  } \hspace{1cm} \hat{I}_{N} =  \frac{1}{N}\sum_{i=1}^{N} f(\boldsymbol{x}_{i}),
\end{equation}

where the  points $\boldsymbol{x}_{1},\ldots,\boldsymbol{x}_{N}$ are sampled from the unit hypercube $[0,1)^{s}.$ The accuracy of the approximation depends on a) how these $N$ points are chosen and b) on the smoothness properties of the function $f$. For a given class of functions $f$, how fast does $E_{N} = |I - \hat{I}_{N}| \rightarrow 0$ depends on how the points are distributed in the space. One way to quantify the spread of points is by using the concept of \emph{discrepancy}. The term \emph{low discrepancy sequences} is used to refer to the sets/sequences of points sampled from a space which try to mimic the uniform distribution over that space. That is the empirical distribution function $\hat{F}(x)$  of these points closely matches the uniform distribution function $F(x)$ on that space and the term \emph{discrepancy} refers to the difference between the two. One of the most commonly used measure of discrepancy is the \emph{star discrepancy} which is given by

\begin{equation*}
D^{*}(\mathcal{P}_{N}) = \sup_{\boldsymbol{x}\in[0,1)^{s}} |F(\boldsymbol{x}) - \hat{F}(\boldsymbol{x})|,
\end{equation*}

where $\mathcal{P}_{N}$  denotes a sample of size $N$ over the unit hypercube. Note that the star discrepancy is the multivariate extension of the Kolmogorov-Smirnov statistic. A set of points $\mathcal{P}_{N}$ is called a \emph{low discrepancy sequence} if $D^{*}(\mathcal{P}_{N}) \in O(N^{-1}(\log N)^{s}).$   Integration methods which use rules  such as (\ref{lds1}) that use deterministic point sets  (e.g. LDS points) are referred to as \emph{quasi-Monte Carlo} (QMC) integration. An important result is the \emph{Koksma-Hlawka inequality} which states that for the QMC integration, if the function $f$ has a total variation in the sense of Hardy and Krause $V(f)$ that is finite then we have an upper bound on the absolute error of integration given by

\[
E_{N} \leq D^{*}(\mathcal{P}_{N}) V(f).
\]

Thus, for functions with bounded variations, the integration error for QMC integration using LDS points is bounded above by $O((\log N)^{s} \times N^{-1})$. On the other hand, the (probabilistic) Monte Carlo (MC) integration error is in $O(N^{-1/2}).$  It can be shown (Niederreiter 1992) that for the $n-$point regular grid, where $N=n^{s}$, the star discrepancy is  $D^{*}(\mathcal{P}_{N}) = 1 - (1 - 1/n)^{s}$ and therefore $D^{*}(\mathcal{P}_{N}) \in O(N^{-1/s}),$   whereas the star discrepancy of a random point set  $D^{*}(\mathcal{P}_{N}) \in O( \sqrt{\log\log N}/\sqrt{N})$ with probability $1$. So for $s>2$, the star discrepancy for randomly generated points converges faster to $0$ than that of an $n-$point regular grid.

While the discrepancy of random points will converge to $0$ in probability, they do suffer from large gaps and clusters and this can affect the accuracy of the estimate for a given set of points. LDS points are deterministic and evenly spaced and they do not suffer from this drawback. In practice, the approximations using LDS converge much faster and hence turn out to be more efficient than both a grid as well as the random points.

Low discrepancy point sets/sequences can be broadly classified into two classes based on the way they are constructed, namely: \emph{lattices} and \emph{digital nets and sequences}. Star discrepancy is not the only measure used, similar results can also be obtained using the $L_{2}$ norm, see Hickernell (1998), for example. It is also  possible to construct point sets/sequences which are low discrepancy with respect to a particular probability measure instead of the Lebesgue measure. Although LDS are usually defined over a unit hypercube, a simple linear transformation can be used to define them over a general $[\boldsymbol{a},\boldsymbol{b})$ hypercube. The numerical examples presented in this paper have all been implemented using just one type of an LDS, namely Korobov lattice (see, for example, Sloan and Joe (1994)). For a general introduction to QMC and its applications refer to Lemieux (2009); for a detailed mathematical account of the digital nets and sequences refer to Dick and Pillichshammer (2010).

\subsection{Notation}

Let $\pi(\boldsymbol{\theta})$ be an $s$ dimensional (posterior) distribution that is only known up to the normalising constant. We assume that $\int_{\Theta} \pi(\boldsymbol{\theta})\,d\boldsymbol{\theta}$ is not analytically known, which is a typical situation in the Bayesian analysis. In this paper, for the sake of simplicity, we will not explicitly mention the dependence of the posterior distribution on the observed data $\boldsymbol{y}$. That is we will use $\pi(\boldsymbol{\theta})$ instead of the usual $\pi(\boldsymbol{\theta}|\boldsymbol{y}).$ \noindent Consider the following approximation to this $s$ dimensional integral

\[
\int_{\Theta} \pi(\boldsymbol{\theta})\,d\boldsymbol{\theta} \approx \frac{\prod_{i=1}^{s} (b_{i} - a_{i})}{N} \sum_{j=1}^{N} \pi(\boldsymbol{\theta}_{j}) =  \mathcal{V} \frac{1}{N} \sum_{j=1}^{N} \pi(\boldsymbol{\theta}_{j}),
\]

where, $\Theta = [\boldsymbol{a},\boldsymbol{b})$, $\boldsymbol{a}=(a_1,\ldots,a_s)\in \mathbb{R}^{s}$,  $\boldsymbol{b}=(b_1,\ldots,b_s)\in \mathbb{R}^{s}$,  and $\boldsymbol{\theta}_{j} \in \Theta, \, j=1,\ldots, N,$ is a set of points at which $\pi$ is evaluated. This set of points could be obtained either by sampling randomly (Monte Carlo approach) or using grid points or using an LDS (quasi-Monte Carlo approach) or indeed using any other method. Note that here the term $\mathcal{V} = \prod_{i=1}^{s} (b_{i} - a_{i})$ is required since $\Theta \neq [0,1)^{s}$.

Similarly, the marginal distribution of the $k^{th}$ component $\theta_{k}$ can be approximated as

\begin{eqnarray*}
\nonumber \pi(\theta_{k}) &=& \int_{\Theta \setminus [a_{k},b_{k})} \pi(\boldsymbol{\theta})\,d\boldsymbol{\theta}_{-k} \approx \frac{\prod_{-k} (b_{i} - a_{i})}{N} \sum_{j=1}^{N} \pi(\theta_{k};\boldsymbol{\theta}_{j})\\
&=& \mathcal{V}_{-k} \frac{1}{N} \sum_{j=1}^{N} \pi(\theta_{k};\boldsymbol{\theta}_{j}),
\end{eqnarray*}

where $\mathcal{V}_{-k} =  \prod_{-k} (b_{i} - a_{i}), \;\boldsymbol{\theta}_{j} \in \Theta\setminus [a_{k},b_{k})$. Suppose the marginal distribution of $\theta_{k}$ is to be evaluated at each of the $n$ distinct fixed points $\theta_{k} = \theta_{k_{l}}, l=1,\ldots,n$. Also assume that, for each $\theta_{k_{l}}$, $\pi$ is evaluated at $m$ points in $ \Theta\setminus [a_{k},b_{k})$; thus  $\pi$ is evaluated at  $N = n \times m$ points in total. Then, the marginals evaluated at each of these distinct points can be approximated as

\begin{eqnarray*} \label{pointwise mean}
\nonumber \pi(\theta_{k} = \theta_{k_{l}}) &=& \int_{\Theta \setminus [a_{k},b_{k})} \pi(\theta_1,\ldots,\theta_k = \theta_{k_{l}},\ldots, \theta_s)\,d\boldsymbol{\theta}_{-k} \\ &\approx & \mathcal{V}_{-k}\frac{1}{m} \sum_{j=1}^{m} \pi(\theta_k=\theta_{k_{l}};\boldsymbol{\theta}_{j}) = \hat{\pi}(\theta_{k_{l}}),\, \boldsymbol{\theta}_{j} \in \Theta\setminus [a_{k},b_{k}).
\end{eqnarray*}

$\hat{\pi}(\theta_{k_{l}})$ is the point-wise mean obtained by averaging out over the remaining dimensions at point $\theta_{k_{l}}$.






\section{Approximation to the posterior using deterministic point sets} \label{approxpost}

\subsection{Grid based functional approximations} \label{grid}

Most of the recent grid based approaches  see, for example, Simpson et.\ al.\ (2016), Lindgren and Rue (2011), Joshi (2011), Ormerod (2011), and Austad and Friel (2010) are based on or inspired by the Integrated Nested Laplace Approximation (INLA) proposed by Rue et.\ al.\ (2009). Although the exact details regarding the functional approximation $\hat{\pi}(\boldsymbol{\theta})$ to the posterior $\pi(\boldsymbol{\theta})$ vary in each case, the main idea underlying all these methods can be summarised as follows:

\noindent \begin{underline} {\textbf{Algorithm: Grid based inference} } \end{underline}

\begin{enumerate}

\item{Locate the mode of $\hat{\pi}(\boldsymbol{\theta})$ using a numerical algorithm}

\item{Identify (estimate) the support $\hat{\Theta}$}

\item{Create a grid $\mathcal{G}$ over the support and evaluate $\hat{\pi}(\boldsymbol{\theta}_{j}),\, \forall \boldsymbol{\theta}_{j} \in \mathcal{G}$}

\item{The marginals $\pi(\theta_{k})$ can now be obtained by numerical integration}

\item{$\hat{\pi}(\boldsymbol{\theta}_{j})$ can also be used to approximate posterior marginals for the latent variables}

\end{enumerate}

\begin{figure} [ht]

\includegraphics[scale=0.4]{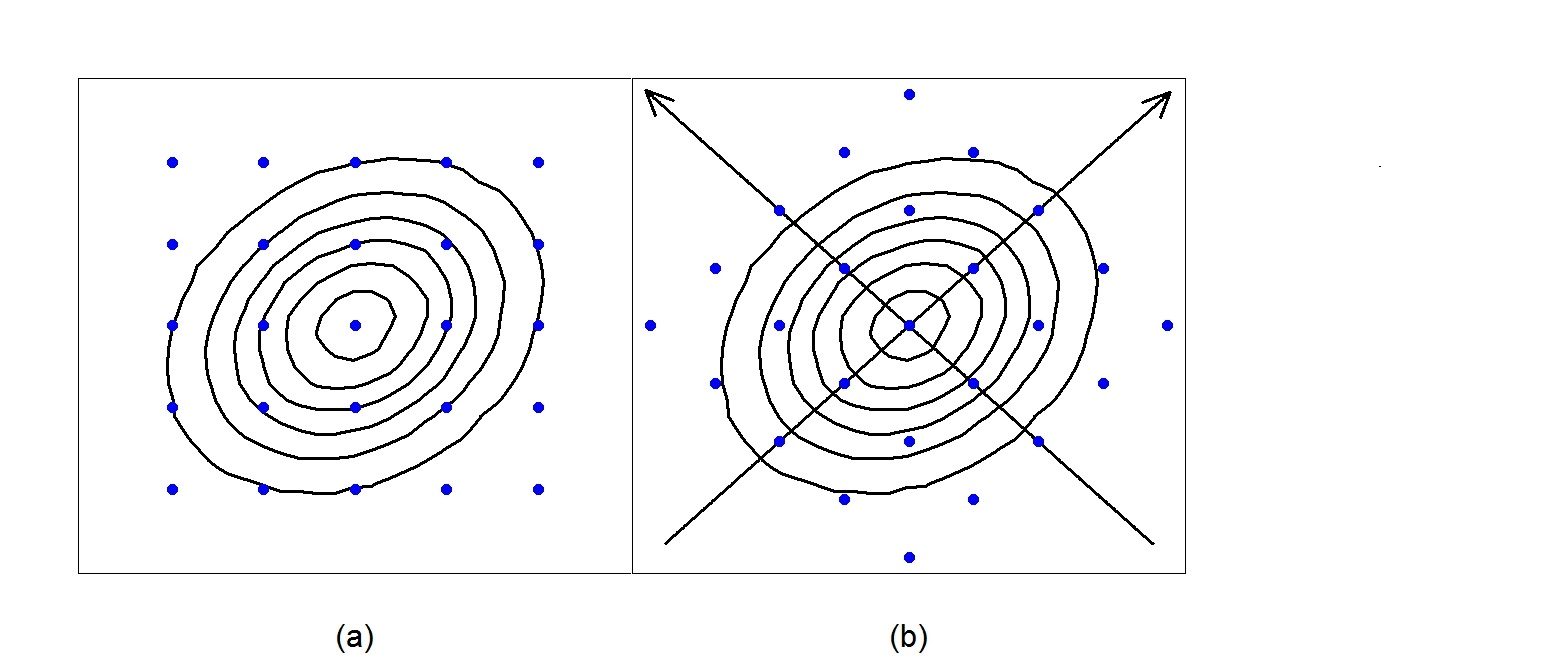}

\caption[]{Two dimensional contour plot with (a) regular grid along the parameter axes  and (b) along the eigen axes.}

\end{figure}

Note that steps $2$ and $3$ are crucial for the accuracy of this algorithm. For unimodal densities, step $2$ is typically implemented by first finding the mode and then the support  by either estimating the standard deviations along each axis using the inverse of the negative Hessian evaluated at the mode or by exploring the density along each eigen axis until the density is negligibly close to $0$.  Such approaches may yield a reasonable approximation of the support but can not guarantee accuracy. Here, however, we do not focus on this problem and assume that accurate support is available. The easiest way to implement step $3$ is to create a regular grid along the axes (Figure $1$ (a)) since the marginals  $\pi(\theta_{k}),\, k=1,\ldots, s,$ could be obtained by simply averaging out over $\theta_{-k}$ and fitting a smoother through those averages. Rue et.\ al.\ (2009) suggest exploring the space along the eigen axes and thus creating a grid along the eigen axes (Figure $1$ (b)) instead of the parameter axes to aid the exploration of the space. While such a point set may resemble a LDS (specifically a lattice), it is not, since the eigen axes are orthogonal to each other. Thus, it does not have any computational benefit over a grid created along the parameter axes, additionally, the marginals can no longer be obtained by simply averaging out and fitting a smoother.

For the remainder of this paper, we may sometimes refer to the $n-$point regular grid (Figure $1$ (a)) as the `$n-$point grid' or simply as `grid' unless we specify otherwise. Here $n$ denotes the number of points along each axis.

\subsection{Using low discrepancy sequences}

In addition to the computational cost that increases exponentially with $s$ which limits the grid based approach to be applicable only when the dimensionality of the parameter space is very small (typically $5$ or less), the other main drawback of this approach is that it fails to accurately capture the shape of the distribution when the distribution is multi-modal or highly skewed unless the number of grid points is large. This is illustrated in Figures $2$  ($4$ dimensional multi-modal distribution) and $3$ ($4$ dimensional Beta distribution). In Figure $2$, a $5-$point grid ($5^4 = 625$ points) fails to capture the shape of the marginal distributions. In this case, one needs at least a $10-$point grid ($10,000$ points) to be able to capture the shape accurately. In Figure $3$, while a $5-$point grid is able to correctly capture marginal $4$ and an $8-$point grid ($4096$ points) is able to accurately capture the symmetric marginal (marginal $2$), it is not able to accurately estimate the remaining  two highly skewed posteriors\footnote{The marginals are estimated by averaging out the remaining dimensions and then fitting a spline - as discussed in Section \ref{grid}.  Approximation using a  least squares polynomial also yields similar results}.

 This happens because a grid is not a \emph{fully projection regular} point set. A point set $\mathcal{P}_{N}$ is said to be fully projection regular if each of its projections is also low discrepancy and contains $N$ distinct points (Sloan and Joe 1994, Lemieux 2009).  Although an $n$-point regular grid has $N=n^{s}$ points in total, the projection on each of its one dimensional marginals only has $n$ distinct points. To the contrary LDS point sets are typically fully projection regular meaning that the projection on each of its marginals also has $N$ distinct points. This allows LDS to capture the shape of the distribution and its marginals more efficiently. This is illustrated in Figure $4$ where a bi-variate Beta distribution contours are shown along with (a) $5-$point grid and (b) Korobov lattice with $32$ points and the true marginals along with the orthogonal projections of the bi-variate Beta distribution at these points in each case.

While the projection regular property helps the point set to efficiently capture the shape of the distribution, the marginals can no longer be computed simply by averaging out (similar to the grid constructed along the eigen axes). Here, we propose a modification to the grid based algorithm which captures the shape of the posterior distribution more efficiently and accurately and uses a least squares method to estimate the marginals.

\newpage

\begin{figure} [t]

\centerline{\includegraphics[scale=0.26]{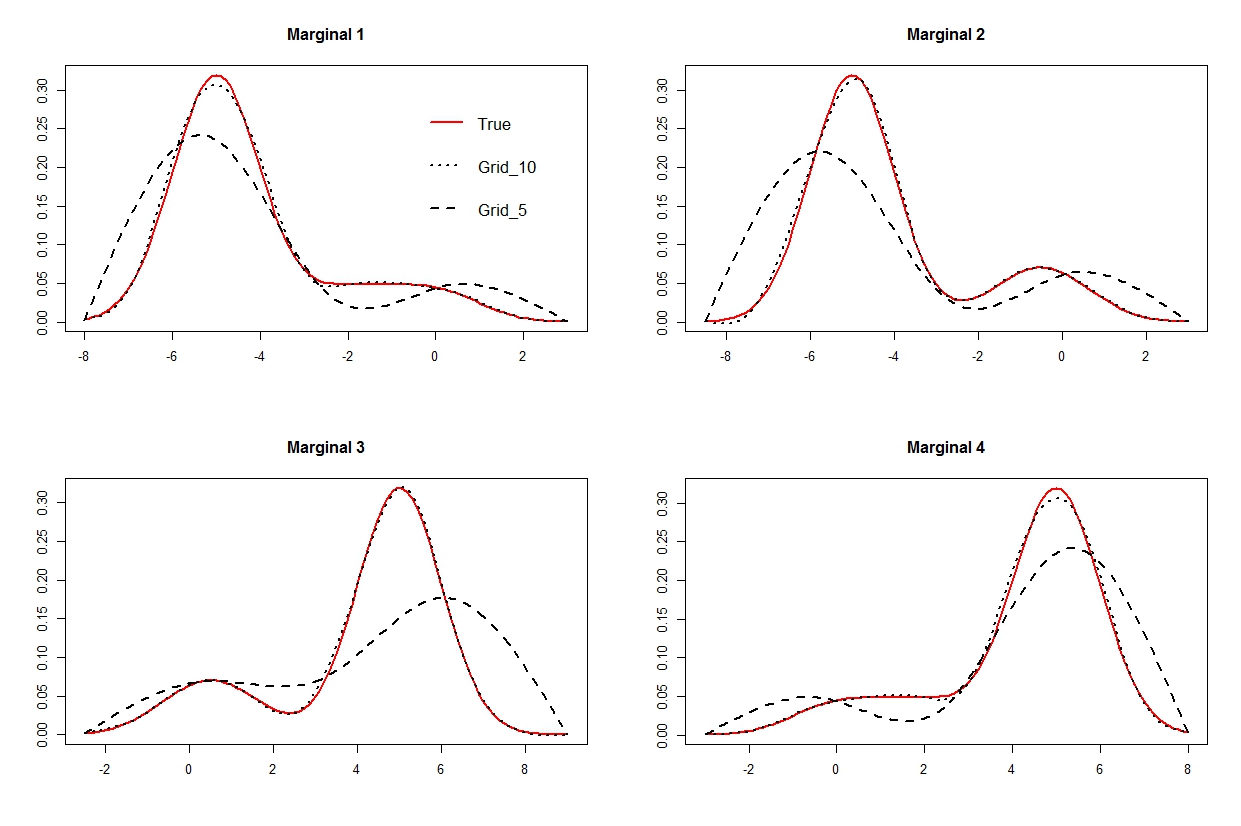}}

\caption[]{Approximating marginals of a four dimensional multi-modal distribution using  a $5-$point grid  and a $10-$point grid.}

\end{figure}

\begin{figure}[ht]

\centerline{\includegraphics[scale=0.26]{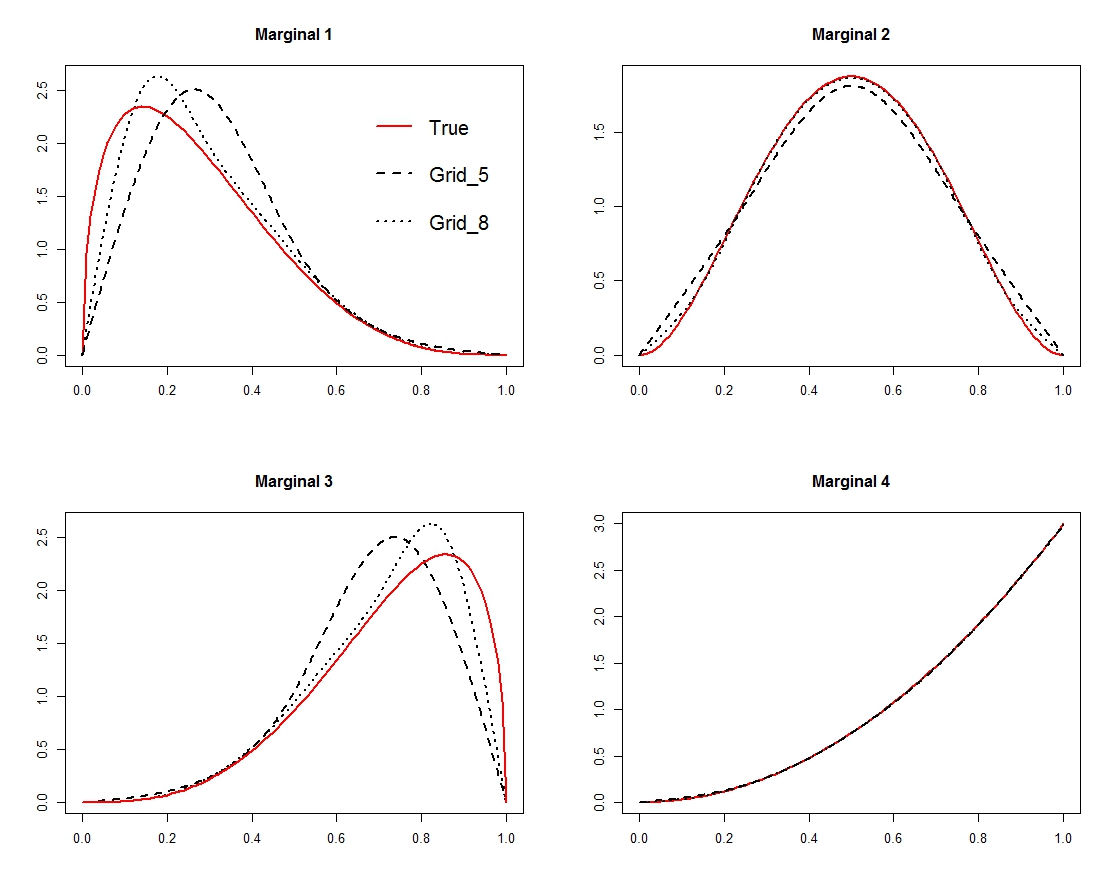}}

\caption[]{Approximating marginals of a four dimensional Beta distribution using  a $5-$point grid  and an $8-$point grid.}

\end{figure}

\clearpage

\vspace{0.5cm}

\noindent \begin{underline} {\textbf{New algorithm: using low discrepancy sequences} } \end{underline}

\begin{enumerate}

\item{Locate the mode of $\hat{\pi}(\boldsymbol{\theta})$ using a numerical algorithm}

\item{Identify (estimate) the support $\hat{\Theta}$}

\item{Generate a low discrepancy point set $\mathcal{P}_{N}$ over $\Theta$ and evaluate $\hat{\pi}(\boldsymbol{\theta}_{j}),\, \forall \boldsymbol{\theta}_{j} \in P_{N}$}

\item{To estimate $\pi(\theta_{k})$: orthogonally project  $\hat{\pi}(\boldsymbol{\theta})$ on $\theta_{k}$ and fit a polynomial of degree $n-1$ to it using the least squares method}

\item{$\hat{\pi}(\boldsymbol{\theta}_{j})$ can also be used to approximate posterior marginals for the latent variables}

\end{enumerate}

Note that the only differences between the grid based algorithm of Section \ref{grid} and the LDS algorithm above are in steps $3$ and $4$. In step $3$, the funcational approximation $\hat{\pi}(\boldsymbol{\theta}_{j})$ is now evaluated on the LDS points instead of the grid points. Thus, for example, if one is using the nested Laplace approximations as in INLA, then one now uses the same approximation but evaluates it on a LDS points set $\mathcal{P}_{N}$ instead of a grid. In step $4$, the marginals $\pi(\theta_{k})$ are now estimated using a least squares method instead of numerical integration.

In Section $3$ we show that marginals obtained using the new algorithm will converge to the true marginals as $N$ increases. More importantly, the results are valid for a wide range of point sets including a  grid, random points as well as LDS points. The numerical results shown in this paper have been obtained using rank-$1$ Korobov lattice rules. For a given $s$, an optimal rank-$1$ Korobov point set can be obtained using the software of L'Ecuyer and Munger (2016).

\begin{figure} [t]
\includegraphics[scale=0.36]{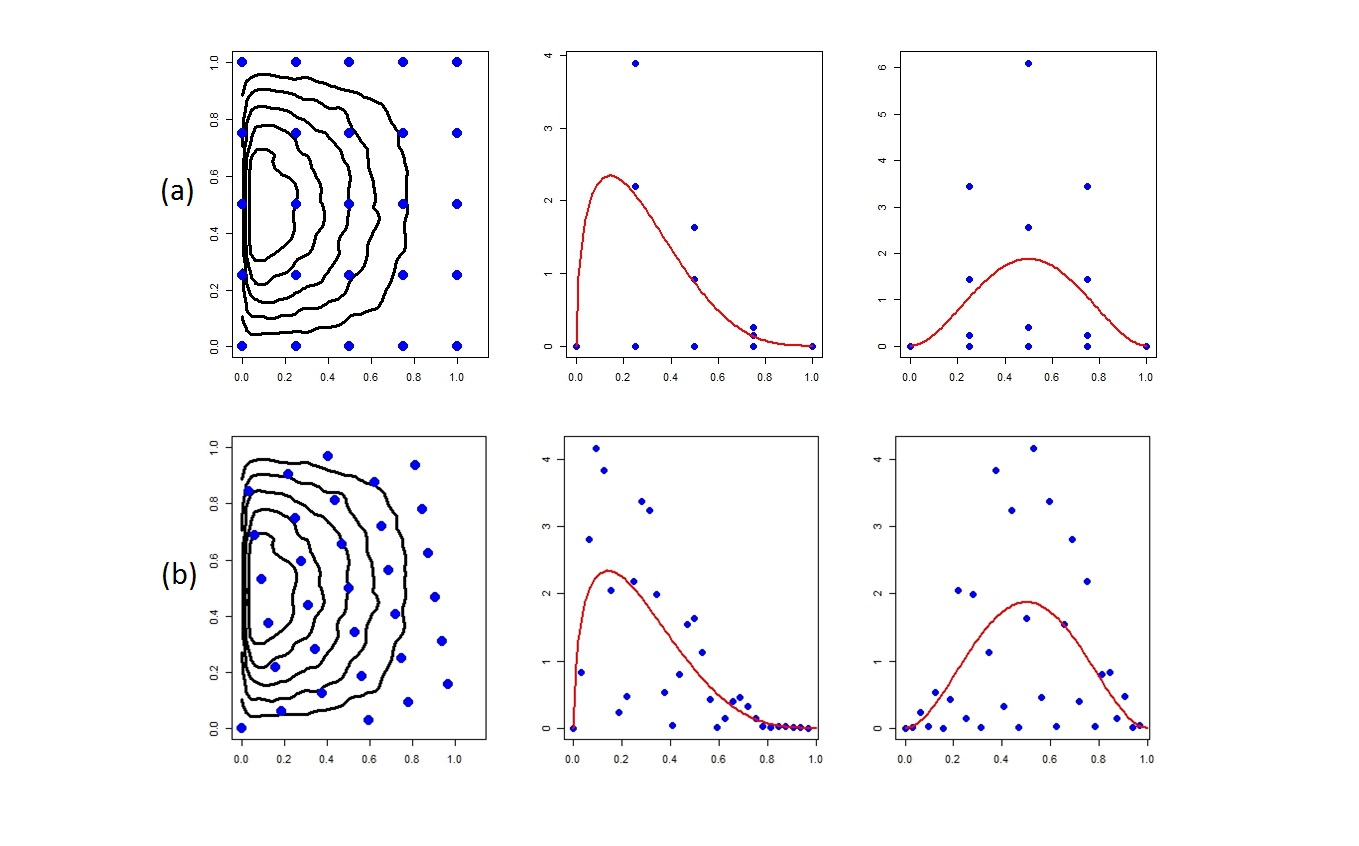}
\caption[]{Bi-variate Beta distribution contours, its true marginals along with orthogonal projections of the bi-variate Beta distribution for (a) $5-$point grid and (b) $32-$point Korobov lattice }
\end{figure}






\section{Convergence theorems} \label{convergence}

The new approach described in the previous section essentially involves evaluating $\pi$ on a set of points in $\Theta$ and then approximating the marginal distribution of $\theta_{k}$ by fitting a least squares polynomial through the orthogonal projections of  $\pi(\boldsymbol{\theta})$ on $\theta_{k}$.  We assume that there are $N=n \times m$ points in $\Theta$ such that, $\pi(\theta_{k})$ is evaluated at $n$ distinct points $\theta_{k_{l}},\, l=1,\ldots,n,$ and that for each unique value of $\theta_{k_{l}}$ there are $m$ points whose $k^{th}$ co-ordinate is equal to $\theta_{k_{l}}$. Note that this description covers a wide variety of point sets. In particular, it includes an $n$-point grid in which case $m=n^{s-1}$. More generally, a point set of this description can be obtained by first fixing the $n$ points $\theta_{k_{l}}$ and then selecting $m$ points $\boldsymbol{\theta}_{j} \in \Theta \setminus [a_{k},b_{k})$ for each distinct value of $\theta_{k_{l}}$ either using random sampling or using an LDS or indeed using any other method. The choice of the points is crucial and determines the convergence properties and the computational efficiency as discussed below.

The orthogonal projections often results in a scatter of points with non-constant variance as can be seen from Figure 4.  This suggests fitting  a weighted least squares polynomial to the orthogonal projection of $\pi(\theta_k=\theta_{k_{l}};\boldsymbol{\theta}_{j})$ on $\theta_{k}$, where the weights are proportional to the variances.  Let the predictions obtained using the weighted least squares polynomial be denoted by $\hat{\pi}_{WLS}(\theta_{k})$. In this case however, it can be shown that if the weighted least squares polynomial is of degree $(n-1)$  then it is equal to the ordinary least squares polynomial of the same degree and therefore it suffices to fit an ordinary least squares polynomial instead.  Let the predictions obtained using the ordinary least squares polynomial be denoted by $\hat{\boldsymbol{\pi}}_{LS}(\theta_{k})$.

Please see Appendix A1 for the details on the orthogonal projection and Appendix A2 for the matrix definitions for the least squares analyses. These will be needed for proving Theorems \ref{thm11} and \ref{thm12} below.

Let $\mathcal{P}_{N}$ be any point set that fits the description above. Then the following theorem holds.

\begin{thm} \label{thm11}

 For any $k \in \{1,\ldots, s\}$, $\hat{\pi}_{WLS}(\theta_{k}) = \hat{\pi}_{LS}(\theta_{k}).$

\end{thm}

\begin{proof}  We have

\begin{eqnarray}  \label{ls1}
\nonumber \hat{\boldsymbol{\pi}}_{LS} (\theta_{k}) &=& \underline{M}(\underline{M}^{T}\underline{M})^{-1} \underline{M}^{T} \boldsymbol{\pi}\\
\nonumber &=& (M \otimes \boldsymbol{1}) \left((M \otimes \boldsymbol{1})^{T} (M \otimes \boldsymbol{1}) \right)^{-1} (M \otimes \boldsymbol{1})^{T} \boldsymbol{\pi}\\
\nonumber &=& (M \otimes \boldsymbol{1}) \left((M^{T} \otimes \boldsymbol{1}^{T}) (M \otimes \boldsymbol{1}) \right)^{-1} (M^{T} \otimes \boldsymbol{1}^{T}) \boldsymbol{\pi}\\
\nonumber &=& (M \otimes \boldsymbol{1}) (M^{T}M \otimes \boldsymbol{1}^{T}\boldsymbol{1})^{-1} (M^{T} \otimes \boldsymbol{1}^{T}) \boldsymbol{\pi}.  \\
\nonumber &=& (M \otimes \boldsymbol{1}) [(M^{T}M)^{-1} \otimes m^{-1}] (M^{T} \otimes \boldsymbol{1}^{T}) \boldsymbol{\pi},\,  \mbox{  since}\boldsymbol{1}^{T}\boldsymbol{1} = m \\ 
\nonumber &=&\frac{1}{m} (M \otimes \boldsymbol{1}) (M^{T}M)^{-1} (M^{T} \otimes \boldsymbol{1}^{T}) \boldsymbol{\pi} \\ 
\nonumber &=& \frac{1}{m} (M(M^{T}M)^{-1}M^{T} \otimes \boldsymbol{1}\boldsymbol{1}^{T}) \boldsymbol{\pi} \\ 
&=& \frac{1}{m} (I_{n} \otimes \boldsymbol{1}\boldsymbol{1}^{T}) \boldsymbol{\pi}, 
\end{eqnarray}

and 

\begin{eqnarray} 
\nonumber \hat{\boldsymbol{\pi}}_{WLS} (\theta_{k}) &=& \underline{M} (\underline{M}^{T}\underline{W}\underline{M})^{-1} \underline{M}^{T} \underline{W} \boldsymbol{\pi}\\
\nonumber &=& (M \otimes \boldsymbol{1}) \left((M \otimes \boldsymbol{1})^{T} (W \otimes I_{m}) (M \otimes \boldsymbol{1}) \right)^{-1} (M \otimes \boldsymbol{1})^{T} (W \otimes I_{m}) \boldsymbol{\pi}\\
\nonumber &=& (M \otimes \boldsymbol{1}) \left((M^{T} \otimes \boldsymbol{1}^{T}) (W \otimes I_{m}) (M \otimes \boldsymbol{1}) \right)^{-1} (M^{T} \otimes \boldsymbol{1}^{T}) (W \otimes I_{m}) \boldsymbol{\pi} \\
\nonumber &=&(M \otimes \boldsymbol{1}) (M^{T}W M \otimes \boldsymbol{1}^{T}I_{m}\boldsymbol{1})^{-1} (M^{T} \otimes \boldsymbol{1}^{T}) (W \otimes I_{m}) \boldsymbol{\pi} \\
\nonumber &=&(M \otimes \boldsymbol{1}) [(M^{T}WM)^{-1} \otimes (m)^{-1}] (M^{T} \otimes \boldsymbol{1}^{T}) (W \otimes I_{m}) \boldsymbol{\pi}, \, \mbox{  since } \boldsymbol{1}^{T}I_{m}\boldsymbol{1} = \boldsymbol{1}^{T} \boldsymbol{1} = m \\ 
\nonumber &=&\frac{1}{m} (M \otimes \boldsymbol{1}) (M^{T}WM)^{-1} (M^{T} \otimes \boldsymbol{1}^{T}) (W \otimes I_{m}) \boldsymbol{\pi} \\ 
\nonumber &=& \frac{1}{m} (M(M^{T}WM)^{-1}M^{T}W \otimes \boldsymbol{1}\boldsymbol{1}^{T}I_{m}) \boldsymbol{\pi} = \frac{1}{m} (I_{n} \otimes \boldsymbol{1}\boldsymbol{1}^{T}) \boldsymbol{\pi}.
\end{eqnarray}

\end{proof}

Let $\hat{\boldsymbol{\pi}}(\theta_{k_{l}})$ be the $m \times 1$ vector where every element is equal to $\frac{\hat{\pi}(\theta_{k_{l}})}{\mathcal{V}_{-k}}$.

\begin{thm} \label{thm12}

 For any $k \in \{1,\ldots, s\}$, $\mathcal{V}_{-k} \hat{\boldsymbol{\pi}}_{LS}(\theta_{k}) = \hat{\boldsymbol{\pi}}(\theta_{k_{l}}).$

\end{thm}

\begin{proof} Using Equation  (\ref{ls1}) we have that

\begin{eqnarray*}
\nonumber \hat{\pi}_{LS} (\theta_{k}) &=& \frac{1}{m} (I_{n} \otimes \boldsymbol{1}\boldsymbol{1}^{T}) \boldsymbol{\pi}\\
&=& \frac{1}{m}
\begin{pmatrix}
    J_{m} & 0_{m} &  \dots  & 0_{m} \\
    0_{m} & J_{m} & \dots  & 0_{m}\\
    \vdots & \vdots  & \ddots & \vdots \\
    0_{m} & 0_{m} & \dots  & J_{m} \\
\end{pmatrix}
\begin{pmatrix}
\boldsymbol{\pi}_1 \\
\boldsymbol{\pi}_2\\
\vdots\\
\boldsymbol{\pi}_n\\
\end{pmatrix} 
=
\begin{pmatrix}
\hat{\boldsymbol{\pi}}(\theta_{k_{1}}) \\
\hat{\boldsymbol{\pi}}(\theta_{k_{2}})\\
\vdots\\
\hat{\boldsymbol{\pi}}(\theta_{k_{n}})\\
\end{pmatrix},
\end{eqnarray*}

where each element $J_{m}$ or $0_{m}$ is a square matrix of size $m \times m$ that contains all 1's or all 0's respectively and  $\boldsymbol{\pi}_{l},\; l=1,\ldots,n$ is the $m \times 1$ vector of function evaluations $\pi(\boldsymbol{\theta})$  corresponding to $\theta_{k_{l}}$.

 \end{proof}

Theorem \ref{thm12} implies that this approach is equivalent to the interpolating polynomial approach where a polynomial of degree $(n-1)$ is fitted to $n$ function evaluations. Therefore the convergence properties can be studied using the relevant literature in Numerical Analysis. For an arbitrary set of fixed points $\theta_{k_{l}}$, the interpolating polynomial does not have good convergence properties in general. However, if $\theta_{k_{l}}$ are chosen either as Chebyshev nodes\footnote{Chebyshev nodes are the roots of the Chebyshev polynomial of the first kind} (shifted to $[a_{k},b_{k})$) or as equally spaced points then the resulting interpolating polynomial will converge to the true function under strong smoothness conditions on the function.  For Theorems \ref{thm2} and \ref{thm3} we now assume that $\theta_{k_{l}}$ are chosen either as Chebyshev nodes or as equally spaced points respectively, points $\boldsymbol{\theta}_{j} \in \Theta \setminus [a_{k},b_{k})$ could be sampled either using a grid or randomly or using an LDS. A special case is where points $\boldsymbol{\theta}_{j} \in \Theta$ are obtained using a grid of $n$ Chebyshev nodes on [$a_{k},b_{k})$ for each $k$ or using a grid of $n$ equally spaced points on [$a_{k},b_{k})$ for each $k$.

\begin{thm} \label{thm2}

If $\pi(\theta_{k})$ is infinitely differentiable such that $$\max_{\xi \in [a_{k},b_{k})} |\pi^{(n)}(\xi)| \leq C,\, \forall n,$$  for some $ C<\infty$ such that $ \frac{C}{2^{(n-1)}}\left(\frac{b_{k}-a_{k}}{2}\right )^{n} \ll (n-1)! ,\, \forall n$, and $\theta_{k_{l}}$ correspond to Chebyshev nodes on the interval $[a_{k},b_{k})$, then  $\mathcal{V}_{-k} \hat{\pi}_{LS}(\theta_{k}) \rightarrow \pi(\theta_{k})$ as $m\rightarrow \infty$ and $n\rightarrow \infty$.

\end{thm}

\begin{proof} 

As $m\rightarrow \infty$, 

\begin{equation}
\hat{\pi}(\theta_{k_{l}}) = \frac{\prod_{-k} (b_{i} - a_{i})}{m} \sum_{j=1}^{m} \pi(\theta_k=\theta_{k_{l}};\boldsymbol{\theta}_{j}) \rightarrow \pi(\theta_{k} = \theta_{k_{l}}).
\label{eqn_converge}
\end{equation} 

Equation (\ref{eqn_converge}) holds due to Koksma-Hlawaka inequality if the $\boldsymbol{\theta}_{j}$ are sampled using a grid or an LDS and due to the Law of Large numbers if the $\boldsymbol{\theta}_{j}$ are sampled randomly.

Then, from Theorem \ref{thm12} and the standard result in approximation theory (see for example, Cheney and Kincaid (1999), Kress (1998)), it can be seen that  $$ \max_{\theta_{k} \in [a_{k},b_{k})} |\pi(\theta_{k}) - \hat{\pi}_{LS}(\theta_{k})| \leq  \max_{\theta_{k} \in [a_{k},b_{k})} \frac{|\pi^{(n)}(\theta_{k})|}{n!} \max_{\theta_{k} \in [a_{k},b_{k})} \prod_{l=1}^{n} |\theta_{k} - \theta_{k_{l}}|.$$

 This implies that  $$ \max_{\theta_{k} \in [a_{k},b_{k})} |\pi(\theta_{k}) - \hat{\pi}_{LS}(\theta_{k})| \leq  \frac{C}{n!} \max_{\theta_{k} \in [a_{k},b_{k})} \prod_{l=1}^{n} |\theta_{k} - \theta_{k_{l}}|.$$

It can be shown (see for example, Sauer (2012)) that if the points $\theta_{k_{l}}$  correspond to the Chebyshev nodes on $[a_{k},b_{k})$, then 

\[
\max_{\theta_{k} \in [a_{k},b_{k})} \prod_{l=1}^{n} |\theta_{k} - \theta_{k_{l}}| \leq \frac{1}{2^{(n-1)}} \left(\frac{b_{k}-a_{k}}{2}\right )^{n}
\]

and therefore,

\[
\max_{\theta_{k} \in [a_{k},b_{k})} |\pi(\theta_{k}) - \hat{\pi}_{LS}(\theta_{k})| \leq  \frac{C}{2^{(n-1)}n!} \left(\frac{b_{k}-a_{k}}{2}\right )^{n}.
\]

\end{proof}


\begin{thm} \label{thm3} If $\pi(\theta_{k})$ is infinitely differentiable such that $$\max_{\xi \in [a_{k},b_{k})} |\pi^{(n)}(\xi)| \leq C,\, \forall n,$$  for some $C<\infty$ such that $C \left(\frac{b_{k}-a_{k}}{n-1}\right )^{n} \ll 1 , \, \forall n$, and $\theta_{k_{l}}$ are  equidistant points then  $\mathcal{V}_{-k} \hat{\pi}_{LS}(\theta_{k}) \rightarrow \pi(\theta_{k})$ as $m\rightarrow \infty$ and $n\rightarrow \infty$.

\end{thm}

\begin{proof} 

From (\ref{eqn_converge}) and it can be shown (see for example, Cheney and Kincaid (1999)) that if the points $\theta_{k_{l}}$ are equally spaced then

\[
\max_{\theta_{k} \in [a_{k},b_{k})} \prod_{l=1}^{n} |\theta_{k} - \theta_{k_{l}}| \leq \frac{(n-1)!}{4} \left(\frac{b_{k}-a_{k}}{n-1}\right )^{n}
\]

and therefore,

\[
\max_{\theta_{k} \in [a_{k},b_{k})} |\pi(\theta_{k}) - \hat{\pi}_{LS}(\theta_{k})| \leq  \frac{C}{4n} \left(\frac{b_{k}-a_{k}}{n-1}\right )^{n}.
\]

\end{proof}

If the function is $n$ times differentiable then the results  in Theorems \ref{thm2} and \ref{thm3} indicate that interpolation obtained using a polynomial of degree ($n-1$) will still be good as long as the derivatives are sufficiently bounded.

Theorems \ref{thm2} and \ref{thm3} provide the conditions under which $\mathcal{V}_{-k}\hat{\pi_{LS}}(\theta_{k}) \rightarrow \pi(\theta_{k}),$  for grids constructed using either the Chebyshev nodes or the equidistant points. This will require $O(n^{s})$ function evaluations though.

Most statistical distributions are smooth with bounded derivatives and therefore satisfy the smoothness requirements of Theorems \ref{thm2} and \ref{thm3}. In Example $5.1$ we illustrate how the exponential distribution, for example, satisfies these smoothness conditions.

We now show how this convergence can also be obtained by using a LDS of size $N=n\times m$ instead. This will require $O(nm)$ function evaluations where, typically, $m \ll n^{s-1}$. We partition the LDS into $n$ equal parts, each part having $m$ points. Let $[\theta_{k_{u}},\theta_{k_{u+1}})$ be one such part. Then we have 

\begin{equation}
\frac{1}{(\theta_{k_{u+1}} - \theta_{k_{u}})} \int_{\Theta \setminus [a_{k},b_{k})} \int_{\theta_{k_{u}}}^{\theta_{k_{u+1}}}  \pi(\boldsymbol{\theta}) \,d\theta_{k} \,d\boldsymbol{\theta_{-k}} \approx \frac{\prod_{-k} (b_{i} - a_{i})}{m} \sum_{j=1}^{m} \pi(\boldsymbol{\theta}_{j}),
 \label{eqn_partition}
\end{equation}

where, $\boldsymbol{\theta}_{j} \in \Theta\setminus [a_{k},b_{k}) \times [\theta_{k_{u}},\theta_{k_{u+1}})$ and $u=1,\ldots,n$. Let this approximation, i.e. RHS of Equation (\ref{eqn_partition}) be denoted by $\hat{\pi}(\theta_{k_{u}}).$

Then, as before, we fit a least squares polynomial of degree $(n-1)$ to the orthogonal projection of $\pi(\boldsymbol{\theta}_{j})$ on $\theta_{k}$. Again, let this least squares polynomial be denoted by $\hat{\pi}_{LS}(\theta_{k})$.

\begin{thm} \label{thm4}

 $\mathcal{V}_{-k} \hat{\pi}_{LS}(\theta_{k}) = \hat{\pi}(\theta_{k_{u}}), \mbox{  for  } u=1,\ldots,n.$

\end{thm}

\begin{proof} The proof is similar to that of Theorems \ref{thm11} and \ref{thm12}. 

\end{proof}

Theorem \ref{thm4} implies that in this case too, the approach is equivalent to an interpolating polynomial approach.

Let  $\Delta \theta_{k} = \theta_{k_{u+1}} - \theta_{k_{u}}$. That is $\theta_{k_{u+1}} =  \theta_{k_{u}} + \Delta \theta_{k}$.

\begin{thm} \label{thm5}

$\hat{\pi}(\theta_{k_{u}}) \rightarrow \pi(\theta_{k_{u}})$ as $m \rightarrow \infty$ and $\Delta \theta_{k} \rightarrow 0$. 

\end{thm}

\begin{proof}

 We have that, as $m\rightarrow \infty,$

\begin{equation}
\hat{\pi}(\theta_{k_{u}}) \rightarrow \frac{1}{(\theta_{k_{u+1}} - \theta_{k_{u}})} \int_{\Theta \setminus [a_{k},b_{k})} \int_{\theta_{k_{u}}}^{\theta_{k_{u+1}}}  \pi(\boldsymbol{\theta}) \,d\theta_{k} \,d\boldsymbol{\theta}_{-k}.
\label{eqn_converge2}
\end{equation} 

Equation (\ref{eqn_converge2}) holds due to the Koksma-Hlawaka inequality.

\noindent Note that 

\begin{equation} \frac{1}{\Delta \theta_{k}} \rightarrow \frac{d}{d\theta_{k_{u}}} \mbox{ as }\Delta \theta_{k} \rightarrow 0.
\label{eqn_calc}
\end{equation}

\noindent Therefore as $m \rightarrow \infty$ and $\Delta \theta_{k} \rightarrow 0$,

\[
 \hat{\pi}(\theta_{k_{u}}) \rightarrow \frac{d}{d\theta_{k_{u}}} \int_{\theta_{k_{u}}}^{\theta_{k_{u+1}}} \int_{\Theta \setminus [a_{k},b_{k})}  \pi(\boldsymbol{\theta})  \,d\boldsymbol{\theta}_{-k} \,d\theta_{k},
\]

using (\ref{eqn_converge2}) and (\ref{eqn_calc}) and changing the order of the integration in (\ref{eqn_converge2}). Note that this can be done thanks to Fubini's theorem since we have assumed that $\pi$ is integrable and Lebesgue measure is a $\sigma-$finite measure.

The theorem is then proven since we have that, $$\frac{d}{d\theta_{k_{u}}} \int_{\theta_{k_{u}}}^{\theta_{k_{u+1}}} \int_{\Theta \setminus [a_{k},b_{k})}  \pi(\boldsymbol{\theta})  \,d\boldsymbol{\theta}_{-k} \,d\theta_{k} = \frac{d}{d\theta_{k_{u}}} \int_{\theta_{k_{u}}}^{\theta_{k_{u+1}}} \pi(\theta_{k})  \,d\theta_{k} = \pi(\theta_{k_{u}}),$$ the last equality using the fundamental theorem of calculus. 

\end{proof}

Since the partition is into equal parts, $\theta_{k_{u}}$ are all equally spaced. The following theorem shows that the least squares approximation will converge to the true marginal.

\begin{thm} \label{thm6}

 If $\pi(\theta_{k})$ is infinitely differentiable such that $$\max_{\xi \in [a_{k},b_{k})} |\pi^{(n)}(\xi)| \leq C,\, \forall n,$$   for some $ C<\infty$ such that $C \left(\frac{b_{k}-a_{k}}{n-1}\right )^{n} \ll 1  \, \forall n$, and $\theta_{k_{u}}, \, u=1,\ldots,n,$ are  equidistant points then  $\mathcal{V}_{-k}\hat{\pi}_{LS}(\theta_{k}) \rightarrow \pi(\theta_{k})$ as $m\rightarrow \infty$ and $n\rightarrow \infty$.

\end{thm}

\begin{proof} 

It follows from Theorems \ref{thm3}, \ref{thm4} and \ref{thm5}.  

\end{proof}

Note that if the function is $n$ times differentiable then the results  in Theorem \ref{thm6}  indicate that interpolation obtained using a polynomial of degree ($n-1$) will still be good as long as the derivatives are sufficiently bounded. This approach requires $O(mn)$ function evaluations, where typically $m < n^{(s-1)}$ and therefore this approach is computationally efficient compared to using an $n$ point grid.

Most statistical distributions are smooth with bounded derivatives and therefore satisfy the smoothness requirements of Theorem \ref{thm6}. In Example $5.1$ we illustrate how the exponential distribution, for example, satisfies these smoothness conditions.






\section{Using f-ANOVA decomposition} \label{f-anova}

As discussed in Section \ref{lds}, for LDS points $D^{*}(P_{N}) \in O(N^{-1}(\log N)^{s}).$ As $s$ increases, the discrepancy will converge more slowly. That is for a fixed $N$, as $s$ increases, the discrepancy will worsen. However, it is possible to generate an LDS which is optimised with respect to certain components (dimensions) of the integrand so that the discrepancy is low corresponding to those components. This can be done by generating an LDS on weighted spaces with weights attached to each component reflecting its relative importance. A natural way to do this is to look at the variance contributions of each of the functional components and attach weights proportional to the variance contribution. The variance contributions can be determined using the functional ANOVA (f-ANOVA) decomposition of the integral. See, for example, Lemieux (2009) and  Owen (2003).

\subsection{f-ANOVA} \label{f-anova1}

The functional ANOVA is useful to decompose an $s$ dimensional integrand as a sum of $2^{s}$ components based on each possible subset $\boldsymbol{\theta}_{I} = (\theta_{i_{1}},\ldots,\theta_{i_{d}})$ of variables, where $I=\{i_{1},\ldots,i_{d}\}\subseteq \{1,\ldots,s\}.$ The decomposition of an $s$ dimensional integrand $\pi(\cdot)$  is given by 

$$ \pi(\boldsymbol{\theta}) = \sum_{I\subseteq \{1,\ldots,s\}} \pi_{I}(\boldsymbol{\theta}), $$ where for nonempty subsets we have $$\pi_{I}(\boldsymbol{\theta}) = \int_{\Theta\setminus \Theta^{d}} \pi(\boldsymbol{\theta})\,d\boldsymbol{\theta}_{-I} -\sum_{J\subset I} \pi_{J}(\boldsymbol{\theta}).$$ Here $d=|I|$ and $\Theta^{d}$ is the parameter subspace for the parameters contained in $I$. The ANOVA component $\pi_{\emptyset}(\boldsymbol{\theta})$ is the integral $I(\pi) = \int_{\Theta}\pi(\boldsymbol{\theta})\,d\boldsymbol{\theta}.$

The expected value of each of these components is $0$ and they also have $0$ covariance. These properties imply that the variance of each component is $$\sigma_{I}^{2} =  \int_{\Theta}\pi_{I}^{2}(\boldsymbol{\theta})\,d\boldsymbol{\theta}$$ and we can write Var$(\pi) = $ Var$(\pi(\boldsymbol{\theta})) = \sigma^{2} =\sum_{I}\sigma_{I}^{2}$. Therefore, $$S_{I} = \frac{\sigma_{I}^{2}}{\sigma^{2}} \in [0,1]$$ can be interpreted as a measure of the relative importance of $\pi_{I}$.

The catch is that the f-ANOVA decomposition requires computing several integrals. Typically, the more straightforward of these integrals is 

$I(\pi) = \int_{\Theta}\pi(\boldsymbol{\theta})\,d\boldsymbol{\theta},$ which is precisely the integral we set out to approximate efficiently. f-ANOVA requires calculating another $2^{s}-1$ integrals to find all the components and then $2^{s}$ integrals to compute the variances for each of those components. Some of these integrals could be infinite if the support is infinite and therefore in addition to being computationally expensive it could also yield a meaningless outcome. Despite these drawbacks, f-ANOVA has been successfully used in situations which require repeated application of a complex integral, in particular for those integrands where the lower order components account for most of the variation. These include option pricing and other applications in finance (see Lemieux (2009) for a detailed review, and also, for example, Griebel et.\ al.\ (2013) and Caflisch et.\ al.\ (1997)). While using f-ANOVA to compute the exact variance components in order to increase efficiency in computing Bayesian posterior seems to defy the purpose, we show below that  it may be possible to find a quick approximation to the variance components using \emph{anchored} f-ANOVA over \emph{weighted spaces}.

\subsection{Anchored f-ANOVA over weighted spaces}

Prior to defining the anchored f-ANOVA over weighted spaces, we must first define the \emph{f-ANOVA over weighted spaces} and the \emph{anchored f-ANOVA}. Please see the Appendix  A$3$ for these details. Here, we propose to use the anchored f-ANOVA over weighted spaces to calculate the estimates of variance components. Working over weighted spaces ensures that the variance integrals are finite (and therefore meaningful). Using the anchored version ensures that the analysis is quick thus retaining the computational advantage.

Let $\boldsymbol{c} = (c_1,\ldots,c_s) \in \Theta$ be the anchor point, $f(\boldsymbol{y} | \boldsymbol{\theta})$ be the likelihood function and $g(\boldsymbol{\theta}) = \prod_{k=1}^{s} g(\theta_{k})$ be the prior distribution, as defined in (\ref{prior1}). We are interested in evaluating the integral

\[
I(\pi) = \int_{\Theta}\pi(\boldsymbol{\theta})\,d\boldsymbol{\theta} = \int_{\Theta} f(\boldsymbol{y} | \boldsymbol{\theta}) g(\boldsymbol{\theta}) \,d\boldsymbol{\theta},
\]

that is, integrating the posterior distribution can also be viewed as integrating the likelihood function over the weighted space defined by the prior distribution. When viewed this way, the integrand to be decomposed is $f(\boldsymbol{y} | \boldsymbol{\theta})$. Thus, we can find the important components of $\pi(\boldsymbol{\theta})$ by employing the anchored f-ANOVA approach to $f(\boldsymbol{y} | \boldsymbol{\theta})$ applied w.r.t weights given by $g(\boldsymbol{\theta}).$

\noindent Then $f_{\emptyset} (\boldsymbol{y} | \boldsymbol{\theta})$ is approximated as  $$f_{\emptyset}(\theta_1,\ldots,\theta_s) = f_{\emptyset}(c_1,\ldots,c_s) \prod_{k=1}^{s} g(c_{k}),$$ and the function corresponding to the first component is $$f_{\theta_1}( \theta_1,\ldots,\theta_s) = f(\theta_1,c_2,\ldots,c_s)  \prod_{k=2}^{s} g(c_{k})- f_{\emptyset}(\boldsymbol{\theta}),$$ and in general for any subset $I$, 

\[ 
f_{I}(\boldsymbol{\theta}) = f(\boldsymbol{c}_{-I};\boldsymbol{\theta}_I)  \prod_{k \not\in I}^{s} g(c_{k}) -\sum_{J\subset I} f_{J}(\boldsymbol{\theta}),
\]

where $f(\boldsymbol{c}_{-I};\boldsymbol{\theta}_I)$ represents the the value of $f(\boldsymbol{\theta})$ evaluated at anchor point $\boldsymbol{c}$ except for the variables involved in $I$ and that the weight functions are only included for those variables which are not contained in $I$. The variance of each component can then be approximated as  

\begin{equation} \label{anchored-2}
\sigma_{I}^{2} =  \int_{\mathbb{R}^{s}}f_{I}^{2}(\boldsymbol{\theta}) g(\boldsymbol{\theta})\,d\boldsymbol{\theta}.
\end{equation} 

Then,  $$S_{I} = \frac{\sigma_{I}^{2}}{\sigma^{2}} \in [0,1]$$ can be interpreted as a measure of the relative importance of $\pi_{I}$. Often this is expressed in percentages instead.






\section{Examples}

In this section we illustrate the convergence properties and the computational benefit of the proposed method using a few standard distributions. Wherever possible, we also compare the results against those obtained using a grid.

\subsection{Exponential distribution}

Most distributions used in statistics are smooth and have smooth derivatives. The Exponential distribution is slightly different since the derivative does not exist at zero. However, here we show that it still satisfies the smoothness conditions imposed by Theorem \ref{thm6} (and also for Theorems \ref{thm2} and \ref{thm3}).  Suppose one of the marginal (posterior) distributions is exponential with parameter $\lambda$. Then we have that,

\begin{equation*}
\pi(\theta_{k}) = \lambda e^{-\lambda \theta_{k}};
\end{equation*}

the $n^{th}$ derivative is given by

\begin{equation*}
\pi^{(n)}(\theta_{k}) = (-1)^{n}\lambda^{n+1} e^{-\lambda \theta_{k}},
\end{equation*}

and

\begin{equation*}
\sup_{\theta_{k}} |\pi^{(n)}(\theta_{k})| = \lim_{\theta_{k} \rightarrow 0+} |\pi^{(n)}(\theta_{k})| = \lambda^{n+1}.
\end{equation*}

Note that, here $\Theta = [a,b) = [0,b)$ for some $b<\infty.$ Then, $\exists n'>0$ and $c<1$ such that $\forall n> n'+1, \, \frac{b}{n-1} \leq \frac{1}{nc} < 1.$  Further, for any $\lambda < \infty, \, \exists n'' > n'$ such that, $\forall n>n'', \lambda^{n+1}\left (\frac{1}{nc}\right )^{n} \ll 1.$

Thus, it can be seen that conditions for Theorem \ref{thm6} (and also for Theorems \ref{thm2} and \ref{thm3}) are satisfied and $\mathcal{V}_{-k} \hat{\pi}_{LS}(\theta_{k}|\boldsymbol{Y}) \rightarrow \pi(\theta_{k}|\boldsymbol{Y})$ as $m\rightarrow \infty$ and $n\rightarrow \infty.$ This is illustrated in Figure $5$. Here, the joint distribution is bi-variate and is a product of two Exponential distributions. We find the least squares approximations to the marginals using Korobov lattices with different $n$ and $m$, the convergence is achieved as they both increase.

\begin{figure} [ht]
\centerline{\includegraphics[scale=0.5]{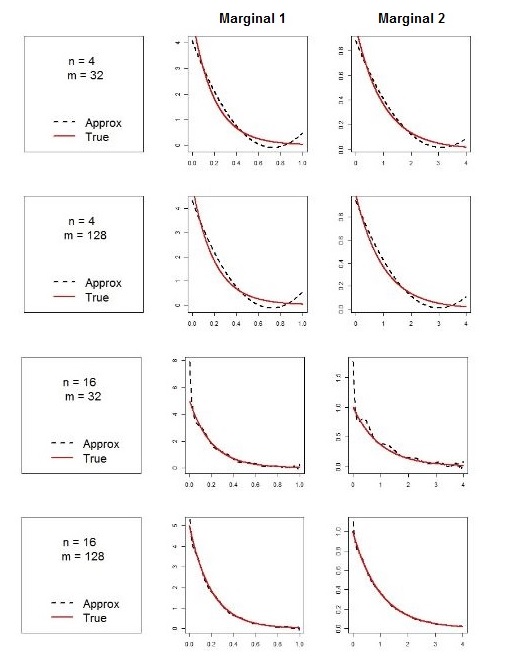}}
\caption[]{ Least squares approximation to the Exponential marginals using Korobov lattices as $n$ and $m$ increase.}
\end{figure}

\subsection{ Multi-modal and skewed distributions}

Figures $2$ and $3$ illustrate that a grid is quite inefficient at accurately capturing the shape of the distribution even in low dimension problems, especially when it is multi-modal or heavily skewed.  Here, we revisit those distributions and try to estimate the marginals using the LDS points to explore the space and fitting the least squares polynomials  of degree $(n-1)$ through the orthogonal projections of the joint distribution on the marginals.

Marginals obtained using the Korobov lattice with $4096$ points in Figure $6$  and with $1024$ points in Figure $7$ show that in both cases, each of the marginals can be very accurately estimated by using much fewer points than the grid. It also shows that in each case using a Korobov lattice with even fewer points still gives estimates better than the ones obtained using a $5-$point grid in Figures $2$ and $3$.

\begin{figure} [ht]
\includegraphics[scale=0.27]{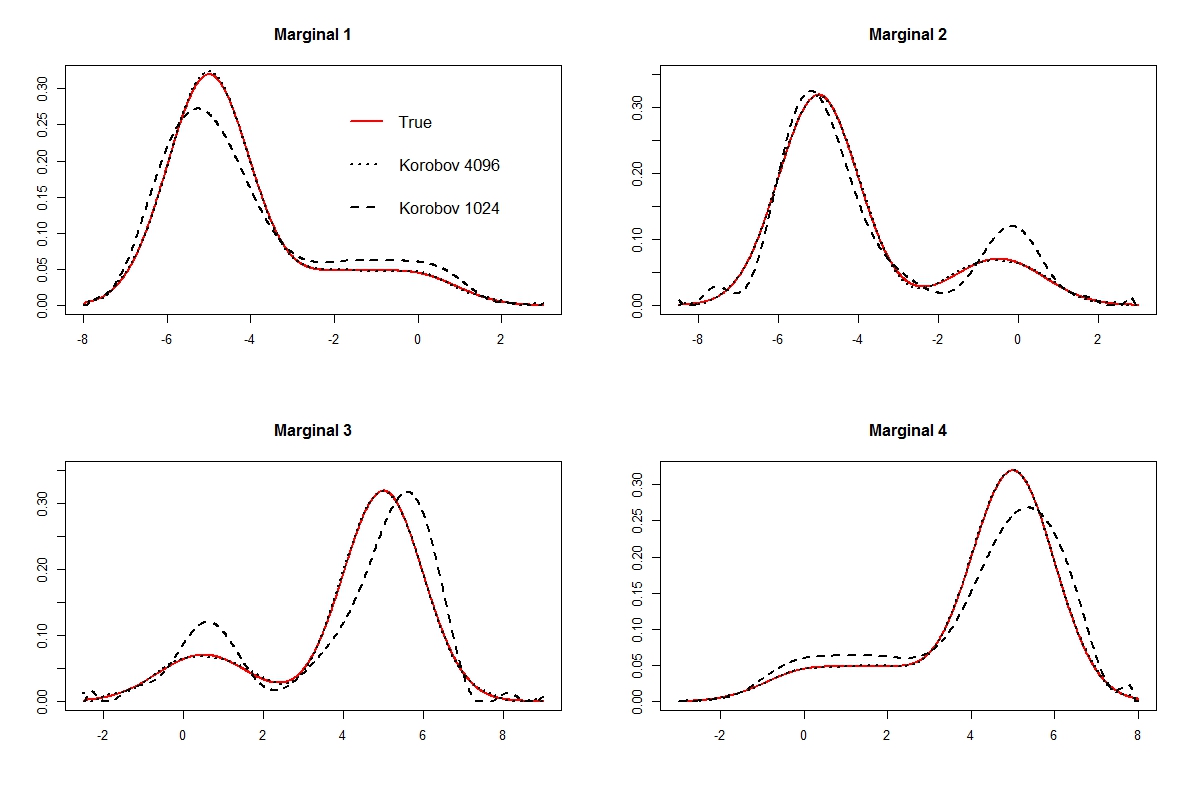}
\caption[]{Approximating marginals of a four dimensional multi-modal distribution using  Korobov lattices with $1024$ points and $4096$ points.}
\end{figure}

\begin{figure} [ht]
\centerline{\includegraphics[scale=0.27]{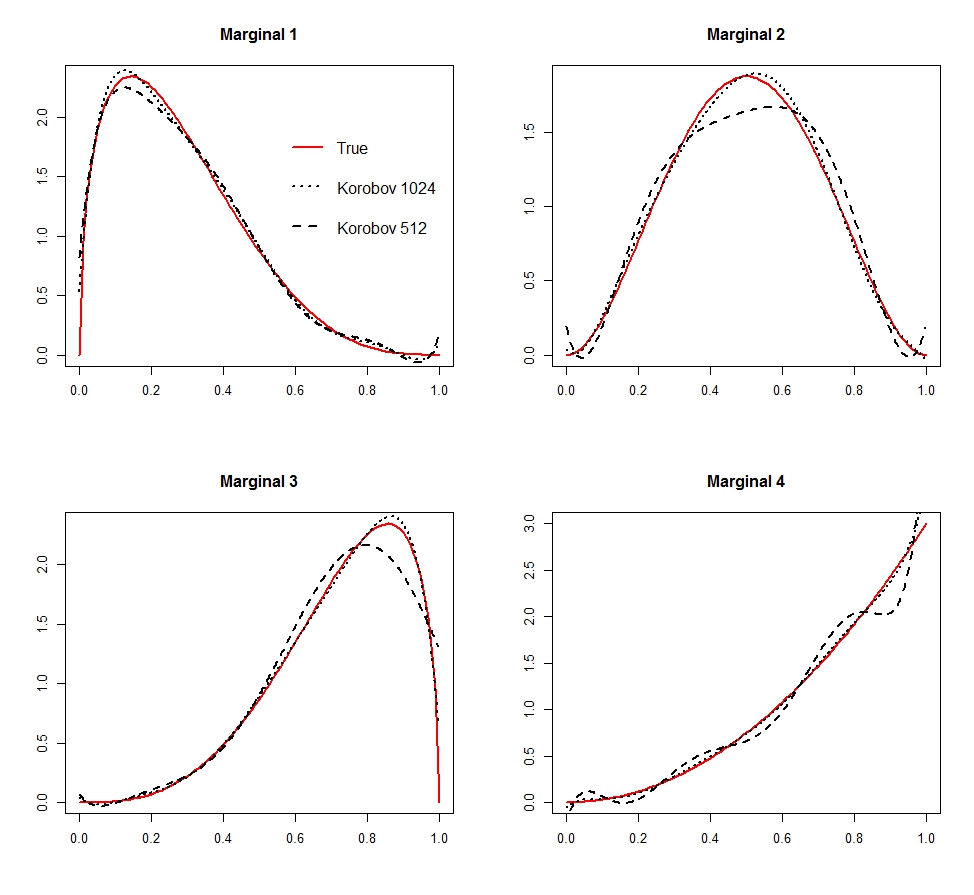}}
\caption[]{Approximating marginals of a four dimensional Beta distribution using  Korobov lattices with $512$ points  and $1024$ points.}
\end{figure}

\subsection{ High dimensional posteriors} \label{gamma-ex}

To illustrate the real computational benefit of using low discrepancy sequences, we consider two posteriors of dimensions $10$ and $12$ respectively. These posteriors have been generated as products of independent Gamma distributions with different parameters. A $5-$point grid will require $5^{10} = 9,765,625$ points in $10$ dimensions and $244,140,625$ points in $12$ dimensions and will likely still yield inaccurate estimates, as illustrated in Figure $3$.

\begin{figure} 
\includegraphics[scale=0.45]{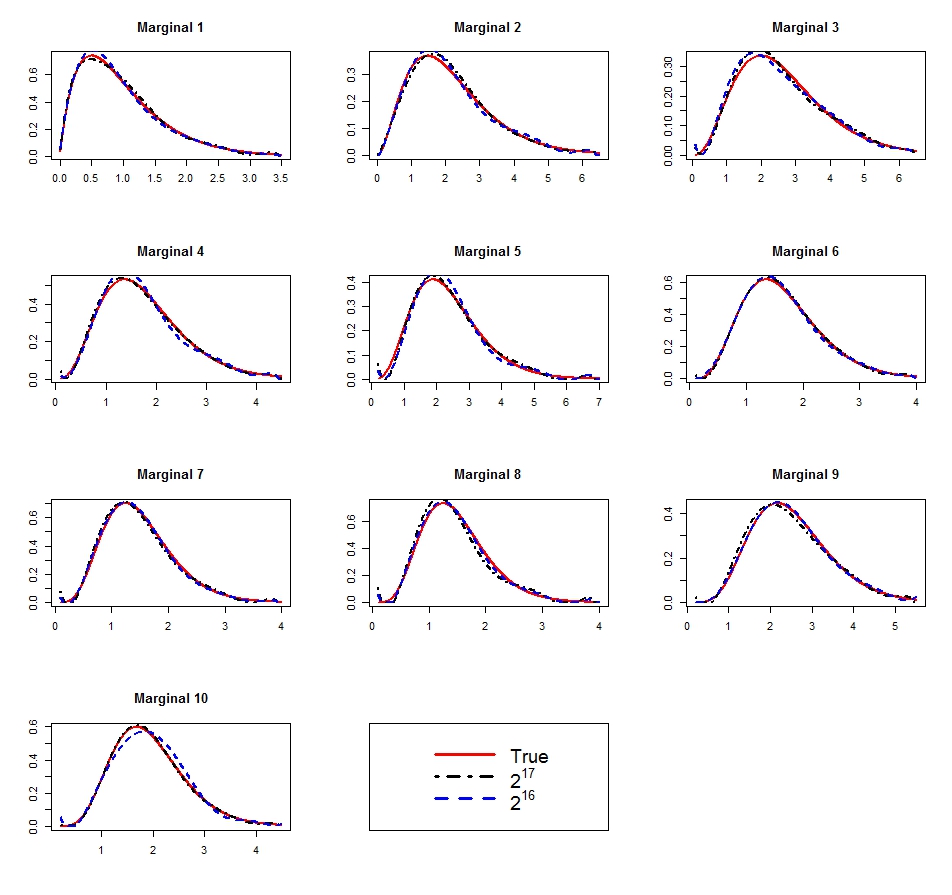}
\caption[]{$10$ dimensional Gamma using Korobov lattice with i) $2^{16} = 65,536$ and ii) $2^{17}=131,072$ points.}
\end{figure}

Figure $8$ shows that for $s=10$, very accurate estimates can be obtained using LDS with as little as $2^{16}$ points ($150$ times fewer than a $5-$point grid). Although estimates obtained using $2^{17}$ points are even more accurate, the difference between the two is very small suggesting that our estimates have started to converge to the true marginals. For $12$ dimensional Gamma, $2^{16}$ points give reasonably accurate estimates and the convergence is achieved by $2^{19} (= 524,288)$ points as can be seen in Figure $9$. However, this is negligible compared to the $244$ million points required for a $5-$point grid.

\begin{figure} 
\includegraphics[scale=0.45]{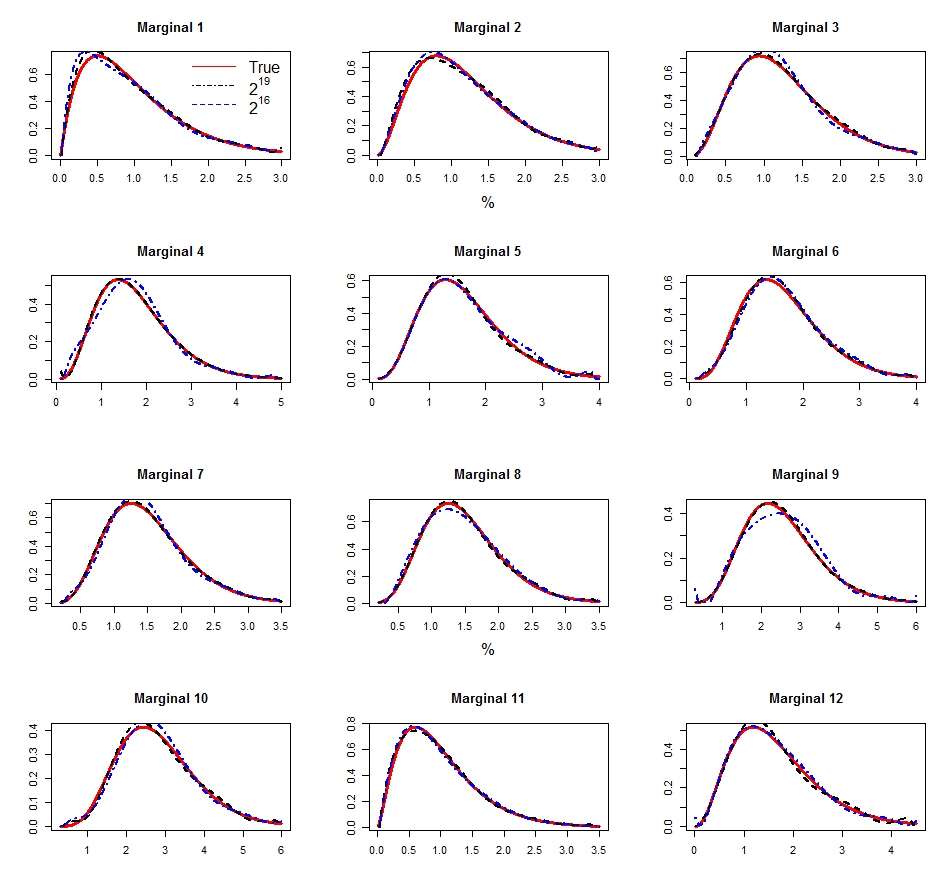}
\caption[]{$12$ dimensional Gamma using Korobov lattice with i) $2^{16} = 65,536$ and ii) $2^{19}=524,288$ points}
\end{figure}

\subsection{High dimensional posteriors using f-ANOVA} \label{f-anova ex}

For some applications each single function evaluation can involve considerable computation and therefore the efficiency gained by using LDS points may not be enough, especially for moderately large parameter spaces. In such cases, f-ANOVA can be used to estimate variance components and generate LDS points according to those weights. We illustrate how this can be done using the anchored f-ANOVA over weighted spaces approach on the Gamma example discussed in Section \ref{gamma-ex}. We assume that the Gamma posteriors have been obtained as a results of using an Exponential likelihood and Gamma priors. We derive the variance components for $2$ dimensional case and then generalise for $s$ dimensions.

Consider two independent Poisson processes with rate parameters $\lambda_1$ and $\lambda_{2}$. Then the waiting times (inter-arrival times) $t_1$ and $t_2$ are Exponentially distributed  with parameters $(\lambda_1)$ and $(\lambda_2)$ respectively. Let $\boldsymbol{t} = (t_1,t_2)$. Then the likelihood is given by

\begin{equation*}
f(\boldsymbol{t}|\lambda_1,\lambda_2) = \lambda_1 \exp[-\lambda_1 t_1] \times  \lambda_2 \exp[-\lambda_2 t_2].
\end{equation*}

The independent conjugate priors for $\lambda_1$ and $\lambda_2$ are $gamma(r_1,v_1)$ and $gamma(r_2,v_2)$ given by

\begin{equation*} g(\lambda_1) \propto \lambda_1^{r_1-1} \exp[-v_1 \lambda_1] \hbox{  and  }  g(\lambda_2) \propto \lambda_2^{r_2-1} \exp[-v_2 \lambda_2]. \end{equation*} Therefore, the bi-variate posterior is the product of two independent gamma distributions

\begin{equation*}
g(\lambda_1,\lambda_2|\boldsymbol{t}) \propto \lambda_1^{r_1}\exp[-(v_1+t_1)\lambda_1] \times \lambda_2^{r_2}\exp[-(v_2+t_2)\lambda_2].
\end{equation*}

Integrating the posterior distribution w.r.t. the Lebesgue measure can be viewed as integrating the likelihood w.r.t. the measure corresponding to the prior distribution. Therefore, we consider anchored functional decomposition of $ f(\boldsymbol{t}|\lambda_1,\lambda_2)$ w.r.t. weights $g(\lambda_1) g(\lambda_2).$

\begin{eqnarray*}
\nonumber f_{\emptyset}(\boldsymbol{t}|\lambda_1,\lambda_2)& = & f(\boldsymbol{t}|c_1,c_2) g(c_1) g(c_2),\\
\nonumber f_{\lambda_1}(\boldsymbol{t}|\lambda_1,\lambda_2) & = &  f(\boldsymbol{t}|\lambda_1,c_2)g(c_2) - f_{\emptyset}(\boldsymbol{t}|\lambda_1,\lambda_2),\\
& = & f(t_1|\lambda_1) f(t_2|c_2) g(c_2)  - f_{\emptyset}(\boldsymbol{t}|\lambda_1,\lambda_2),\\
 f_{\lambda_2}(\boldsymbol{t}|\lambda_1,\lambda_2) &= & f(\boldsymbol{t}|c_1,\lambda_2) g(c_1)- f_{\emptyset}(\boldsymbol{t}|\lambda_1,\lambda_2),\\
\nonumber & = & f(t_1|c_1) f(t_2|\lambda_2) g(c_1)  - f_{\emptyset}(\boldsymbol{t}|\lambda_1,\lambda_2) \hbox{  and  }\\
 f_{\lambda_1 \lambda_2}(\boldsymbol{t}|\lambda_1,\lambda_2) & = &  f(\boldsymbol{t}|\lambda_1,\lambda_2) - f(\boldsymbol{t}|\lambda_1,c_2) g(c_{2})- f(\boldsymbol{t}|c_1,\lambda_2) g(c_1)+ f_{\emptyset}(\boldsymbol{t}|\lambda_1,\lambda_2).
\end{eqnarray*}

For the ease of notation, let the constant terms be denoted as 

\begin{equation*}
a_0 = f(\boldsymbol{t}|c_1,c_2) g(c_1) g(c_2), \; a_1 = f(t_1|c_1)g(c_1) \hbox{ and } a_2 =  f(t_2|c_2) g(c_2).
\end{equation*}

Then,

\begin{equation*}
 f_{\lambda_1 \lambda_2}(\boldsymbol{t}|\lambda_1,\lambda_2) = f(\boldsymbol{t}|\lambda_1,\lambda_2) - f(t_1|\lambda_1) a_2 - f(t_2|\lambda_2) a_1 + a_0.
\end{equation*}

The variance component of $ f_{\lambda_1}(\boldsymbol{t}|\lambda_1,\lambda_2)$ is given by

\begin{eqnarray} \label{f-anova ex1}
\nonumber \sigma_{\lambda_1}^{2} &=& \int_{0}^{\infty}  f_{\lambda_1}^{2} (\boldsymbol{t}|\lambda_1,\lambda_2)\,  g(\lambda_1) g(\lambda_2) \,d(\lambda_1)\, d(\lambda_2)\\
\nonumber & = & \int_{0}^{\infty} [f(t_1|\lambda_1) a_2 - a_0]^{2} \,  g(\lambda_1) g(\lambda_2) \,d(\lambda_1)\, d(\lambda_2).\\
\end{eqnarray}

Solving (\ref{f-anova ex1}) gives

\begin{equation*}
\sigma_{\lambda_1}^{2} = (a_2)^{2} \frac{\Gamma(r_1+2)}{(2t_1+v_1)^{r_1+2}} - 2  a_2 a_0 \frac{\Gamma(r_1+1)}{(t_1+v_1)^{r_1+1}} + (a_0)^{2},
\end{equation*}

and similarly, 

\begin{equation*}
\sigma_{\lambda_2}^{2} = (a_1)^{2} \frac{\Gamma(r_2+2)}{(2t_2+v_2)^{r_2+2}} - 2  a_1 a_0 \frac{\Gamma(r_2+1)}{(t_2+v_2)^{r_2+1}} + (a_0)^{2},
\end{equation*}

\begin{equation*}
\nonumber \sigma_{\emptyset}^{2} = \int_{0}^{\infty}  f_{\emptyset}^{2} (\boldsymbol{t}|\lambda_1,\lambda_2)\,  g(\lambda_1) g(\lambda_2) \,d(\lambda_1)\, d(\lambda_2) = (a_0)^2.
\end{equation*}

In general, for an $s-$dimensional Poisson process, one can show that

\[
f_{\lambda_{k}}(\boldsymbol{t}|\lambda_{1},\cdots,\lambda_{s}) = \prod_{i \ne k} a_{i} [f(t_{k}|\lambda_{k}) - a_{k}], \hbox{ for } k=1,\ldots, s.
\]

and,

\[
\sigma_{\lambda_{k}}^{2} = \prod_{i \ne k} a_{i}^{2}  \left [  \frac{\Gamma(r_k+2)}{(2t_k+v_k)^{r_k+2}}  + a_{k}^{2} - 2 a_{k} \frac{\Gamma(r_k+1)}{(t_k+v_k)^{r_k+1}} \right ].
\]

Note that, here, we only consider the first order components, ignoring the higher order components since they involve solving tedious algebra - this is equivalent to assuming that the hyper parameters are independent of each other. Despite computing approximate variance contributions under this assumption, the low discrepancy sequences generated using these weights improve both the accuracy and the efficiency as illustrated in Figure $10$. Here, we use the same $10$ dimensional Gamma posteriors used in Section \ref{gamma-ex}.

\begin{figure}
\includegraphics[scale=0.45]{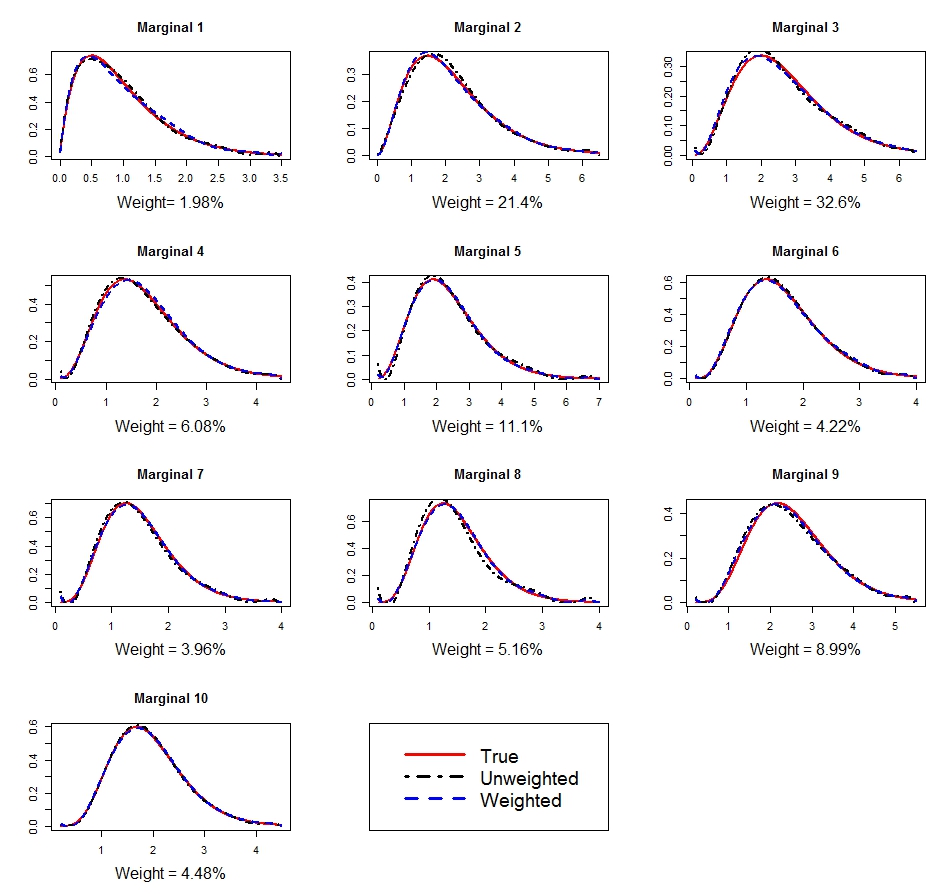}
\caption[]{Marginals of $10$ dimensional Gamma approximated using $2^{17}$ points of a weighted Korobov lattice as well as using $2^{17}$ points of (unweighted) Korobov lattice. The weight associated with each component is given underneath each graph.}
\end{figure}

Figure $10$ shows that the weighted Korobov lattice with $2^{17}$ points yields more accurate estimates of the marginals, especially for marginals with higher contributions to the total variance, than those obtained using the unweighted Korobov lattice with the same number of points. The Hellinger distances between the true and the estimated marginals are given in Table $1$. Another important implication is that, a comparable level of accuracy can be achieved by using fewer number of weighted points as illustrated in Table $1$. It shows that the Hellinger distances between the true and the estimated marginals for a weighted Korbov lattice with only $2^{16}$ points are comparable to (in some cases smaller, in other cases larger than) those obtained for the unweighted Korobov lattice with $2^{17}$ points. Thus, it is possible to further reduce the computational cost by taking into consideration the functional decomposition of the integrand.


%


















\begin{table}

\begin{center}

\begin{tabular}{ c|c|c|c|c}

\hline\\

Marginal & Weights  & Weighted $2^{16}$ & Weighted $2^{17}$ & Unweighted $2^{17}$ \\ [0.5ex]

\hline

1 & 1.98 & 0.05869 & 0.02223 & 0.03056 \\

2 & 21.4 & 0.03525 & 0.01122 & 0.02890\\

3 & 32.6 & 0.03848 & 0.02261 & 0.03731\\

4 & 6.08 & 0.04275 & 0.01904 &  0.03097 \\

5 & 11.1 & 0.04960 & 0.02183 & 0.08474\\

6 & 4.22 & 0.05409 & 0.02886 & 0.03016\\

7 & 3.96 & 0.06880 & 0.03650 & 0.06513\\

8 & 5.16 & 0.08478 & 0.04766 & 0.09922\\

9 & 8.99 & 0.07038 & 0.03187 & 0.05135\\

10 & 4.48 & 0.04998 & 0.02465 & 0.02642\\

\end{tabular}

\caption{Hellinger distances between the true and the estimated marginal using weighted and unweighted Korobov lattices}

\end{center}

\end{table}

\section{Discussion}

This paper aims to make three important contributions. First, we propose that LDS should be used in place of the grid based point sets when exploring the Bayesian posterior using a set of deterministic points.  This  improves computational efficiency and thus creates the possibility to use models with higher number of hyper parameters in this kind of inference framework. Second, we propose a new method to estimate marginal distributions using least squares polynomial fit; not only is this method easy and computationally cheap but also, it can be used on a wide variety of point sets including the grid and random points. Importantly it provides the means to mathematically prove the convergence properties of this approximation for various point sets including for the grid. Despite the recent popularity of grid based methods, this has not yet been done to the best of our knowledge. Third, we propose the use of the f-ANOVA technique to further improve the accuracy and the efficiency and show how this can be done using the anchored f-ANOVA on weighted spaces approach.

Low discrepancy points explore the parameter space more efficiently and have good convergence properties. As a result using these points provides not only the computational benefit but also improves the accuracy. We have shown that, unlike the grid based approaches, using LDS provides an accurate and efficient estimation even when the distribution is highly skewed or multi-modal.  They are also easy to simulate  using various platforms, see for example, L'Ecuyer and Munger (2016) and Christophe and Petr (2015). The main implication of this is that computationally efficient and accurate alternatives to MCMC methods can now be developed even when the (hyper) parameter space has more than $5$ dimensions. While we have illustrated this using $10$ and $12$ dimensional posteriors, posteriors of even higher dimensions can be considered in principle. The only limiting factor is the amount of computational workload required for evaluation of the posterior at every single point and this will vary from application to application.

The new least squares based approach to estimate marginal distributions can be used with a wide range of point sets including the grid and the random points. The key is to fit the polynomial of degree $(n-1)$ if an $n-$point grid is used. Theorems \ref{thm11} and \ref{thm12} show that when exploring $\Theta$ using any $n- $point grid, the least squares polynomial will pass through the point wise means obtained by averaging out the distribution over the remaining $s-1$ dimensions. Theorems \ref{thm2} and \ref{thm3} show that if the grid was constructed using Chebyshev nodes or using equidistant points then the least squares polynomial will converge to the true marginal as $m, n \rightarrow \infty$. However, using a grid requires $O(n^{s})$ function evaluations. Efficiency can be achieved by instead using LDS points and partitioning into $n$ equal parts. Theorem \ref{thm4} shows that in this case too the least squares polynomial will pass through the point wise means obtained for each part. Theorems \ref{thm5} and \ref{thm6} show that the least squares polynomial will in fact converge to the true marginal as $m, n \rightarrow \infty$. Using an LDS requires $O(nm)$ function evaluations, where, typically $m\ll n^{s-1}$ and hence the efficiency is achieved.

It is important to note that in practice the orthogonal projections (on each marginal)  do not have to be performed. One simply fits the least squares polynomial to $\pi(\boldsymbol{\theta})$ regressed against $\theta_{k}$. Also, multiplying the least squares polynomial with volume $\mathcal{V}_{-k}$ is not necessary if the marginals are to be normalised since this term will cancel out during normalisation.

We also show that further efficiency and accuracy can be achieved by taking into account the functional decomposition of the integrand. 

The f-ANOVA approach has not yet been widely explored in the statistical literature. While the f-ANOVA approach can be analytically quite demanding and computationally very expensive, we have shown that it is possible to find quick and useful approximations to the variance components using the anchored f-ANOVA on weighted spaces. The f-ANOVA computations shown in Section \ref{f-anova ex} were relatively simple since we had chosen a closed form likelihood and a conjugate prior distribution. The algebra involved may become quite tedious if the prior distributions are not conjugate and/or the posterior (or its approximation) is not available in a closed form. However, it is important to note that in such cases the variance components can still be computed by using the MC or the QMC approximations to Equation (\ref{anchored-2}) which is only a $1$ dimensional integral when we only consider the first order approximations.  We have shown that using the functional decomposition can improve both the accuracy as well as the efficiency of the estimate and hope that in future, complex statistical models will benefit from this approach.

Convergence results proved here assume that the support is correctly identified and the approximation is exact. In practice, often it is not possible to guarantee that the support has been identified correctly and this could induce some inaccuracy in the estimation. Also, the true posteriors may not be available/known and hence often an approximation to the posterior is evaluated at each point, and this will induce further inaccuracy depending on how good the approximation is. Identifying the support may involve finding the mode/s using numerical methods. In some cases, finding the mode can itself be computationally expensive and the computational cost may increase exponentially with the dimension of the parameter space. Where this is the case, efficient exploration of space using low discrepancy points will be even more important. Points, which are low discrepancy with respect to a particular probability measure (as opposed to the Lebesgue measure) can also be generated.  Thus, it may be possible to use LDS points generated with respect to the posterior distribution. Such an important sampling like approach is likely to further improve the efficiency. However, generating such points gets progressively difficult as the dimensionality increases and therefore this option needs to be explored further.






\appendix

\section*{ A1. Orthogonal projection}

Let  $\pi(\boldsymbol{\theta})$ be an $s$ dimensional distribution that is evaluated at $N$ distinct LDS points $\boldsymbol{\theta}={\boldsymbol{\theta}_{1},\ldots, \boldsymbol{\theta}_{N}}$ in $\Theta \subseteq \mathbb{R}^{s}$, where each $\boldsymbol{\theta}_{j} \in \Theta,\, j=1,\ldots,N,$ is an $s-tuple$ $\boldsymbol{\theta}_{j}=(\theta_{1j},\ldots,\theta_{sj}).$ These points along with the function evaluations, that is $(\boldsymbol{\theta}_{j}, \pi(\boldsymbol{\theta}_{j}))$ are $(s+1)-tuples$ conveniently represented in a matrix form as

\begin{eqnarray*}
\boldsymbol{\Psi}_{(s+1)\times N} = 
\begin{bmatrix}
    \theta_{1,1}  & \dots  & \theta_{s,1} & \pi(\boldsymbol{\theta}_1)   \\
    \theta_{1,2}  & \dots  & \theta_{s,2} & \pi(\boldsymbol{\theta}_2)  \\
    \vdots  & \ddots & \vdots & \vdots\\
    \theta_{1,N}  & \dots  & \theta_{s,N} & \pi(\boldsymbol{\theta}_N)  \\
\end{bmatrix}^\top \in \mathbb{R}^{s+1}
\end{eqnarray*}

To estimate the $k^{th}$ marginal $\pi(\theta_{k})$, we first orthogonally project $ \pi(\boldsymbol{\theta}_{j})$ on the $ k^{th}$ marginal to obtain 

\begin{eqnarray*}
\psi_k = 
\begin{bmatrix}
    \theta_{k,1} & \pi(\boldsymbol{\theta}_1)   \\
    \theta_{k,2} & \pi(\boldsymbol{\theta}_2)  \\
    \vdots   & \vdots\\
    \theta_{k,N} & \pi(\boldsymbol{\theta}_N) \\
\end{bmatrix}^\top \in \mathbb{R}^{2},
\end{eqnarray*}

$\psi_{k} = P_{k}\boldsymbol{\Psi}$, where $P_{k} = A(A^\top A)^{-1}A^\top$ is a projection matrix and $A_{(s+1)\times 2}$ is a unit basis vector for $\mathbb{R}^{2}$ with the $k^{th}$ entry in the first column and the $(s+1)^{th}$ entry in the second column  as one, all the remaining entries are zeros. For example, if $s=3$ and $k=2$ then,

\begin{eqnarray*}
\boldsymbol{\Psi}_{4\times N} = 
\begin{bmatrix}
    \theta_{1,1}  & \theta_{2,1} & \theta_{3,1} & \pi(\boldsymbol{\theta}_1)   \\
    \theta_{2,1}  &  \theta_{2,2}  & \theta_{3,2} & \pi(\boldsymbol{\theta}_2)  \\
    \vdots  & \ddots & \vdots & \vdots\\
    \theta_{1, N}  &  \theta_{2,N}  & \theta_{3,N} & \pi(\boldsymbol{\theta}_N)  \\
\end{bmatrix}^\top, \;
A_{4 \times 2} =  
\begin{bmatrix}
0 & 0\\
1 & 0 \\
0 & 0 \\
0 & 1\\
\end{bmatrix}, \;
P_{4 \times 4} = 
\begin{bmatrix}
0 & 0 & 0 & 0\\
0 & 1 & 0 & 0\\
0 & 0 & 0 & 0\\
0 & 0 & 0 & 1\\
\end{bmatrix}
\end{eqnarray*}

and
 
\begin{eqnarray*}
P \Psi = 
\begin{bmatrix}
   0 &  \theta_{2, 1}  & 0 & \pi(\boldsymbol{\theta}_1)   \\
   0 &  \theta_{2,2} & 0 &\pi(\boldsymbol{\theta}_2)  \\
   \vdots & \vdots & \vdots   & \vdots\\
    0 & \theta_{2,N} & 0 & \pi(\boldsymbol{\theta}_N) \\
\end{bmatrix}^\top, \mbox{ ignoring the rows with zeros }
\begin{bmatrix}
    \theta_{2,1} & \pi(\boldsymbol{\theta}_1)   \\
    \theta_{2,2} & \pi(\boldsymbol{\theta}_2)  \\
     \vdots  & \vdots\\
    \theta_{2,N} & \pi(\boldsymbol{\theta}_N) \\
\end{bmatrix}^\top = \psi_2.
\end{eqnarray*}

\section*{A2. Matrix definitions for least squares analysis}

Let $\mathcal{P}_{N}$ be a point set that fits the description given at the beginning of Section \ref{convergence}.  Let $\underline{M}$ be the design matrix when fitting a least squares polynomial of degree $(n-1)$ through the orthogonal projections of $\pi(\boldsymbol{\theta})$ on $\theta_{k}$. Such a projection has $n$ unique abscissa points $\theta_{k_{l}},\, l=1,\ldots,n.$ Then $\underline{M}$ is of size $N \times n$, and has a block structure,

\[ 
\underline{M} = 
\begin{bmatrix}
    \boldsymbol{1} & \boldsymbol{t}_1 & \boldsymbol{t}_{1}^{2} & \dots  & \boldsymbol{t}_{1}^{n-1} \\
 \boldsymbol{1} & \boldsymbol{t}_{2} & \boldsymbol{t}_{2}^{2} & \dots  & \boldsymbol{t}_{2}^{n-1} \\
    \vdots & \vdots & \vdots & \ddots & \vdots \\
    \boldsymbol{1} & \boldsymbol{t}_{n} & \boldsymbol{t}_{n}^{2} & \dots  & \boldsymbol{t}_{n}^{n-1} \\
\end{bmatrix},
\]

where each element block $\boldsymbol{t}_{l}^{p} \in \underline{M}, (p = 0,\ldots,n-1)$ is an $m \times 1$ column vector containing only the element $\theta_{k_{l}}^{p}$. We can also express $\underline{M}$ as a Kronecker product of the Vandermonde matrix $M$ and the $m \times 1$ column vector of $1's$,

\[
\underline{M} = M \otimes \boldsymbol{1}_{(m \times 1)},
\]

where, $M$ is a squares Vandermonde matrix of size $n$, which is of full rank and is invertible since all elements $\theta_{k_{l}}$ are unique.

For weighted least squares, we assign a weight $w_{l}$ to all projections corresponding to a unique abscissa point $\theta_{k_{l}}.$  We define the weights matrix $\underline{W}$ of size $N \times n$ by

\[
\underline{W} = 
\begin{pmatrix}
w_1I_{m} &0I_{m}&\cdots&0I_{m} \\
0I_{m}&w_2I_{m}&\cdots&0I_{m}\\
\vdots&\vdots&\ddots&\vdots\\
0I_{m}&0I_{m}&\cdots&w_nI_{m}
\end{pmatrix},
\]

where, $I_{m}$ is the identity matrix with size $m \times m$. $\underline{W}$ can also be expressed as a Kronecker product

\[
\underline{W} = W \otimes I_{m},
\]

where $W$ is the $n \times n$ diagonal matrix of weights

\[
W = 
\begin{pmatrix}
w_1&0&\cdots&0\\
0 &w_2&\cdots&0\\
\vdots&\vdots&\ddots&\vdots\\
0&0&\cdots&w_n
\end{pmatrix}.
\]

Let $\boldsymbol{\pi}$ be the $N \times 1$ vector of function evaluations $\pi(\boldsymbol{\theta}_{1}), \ldots,  \pi(\boldsymbol{\theta}_{N}).$

\section*{A3. Details on f-ANOVA decomposition}

\subsection*{A3.1. f-ANOVA over weighted spaces}

f-ANOVA on weighted spaces is defined as follows (see, for example, Griebel et.\ al.\ (2013)).  Let $g$ be a continuous and strictly positive univariate probability density function, i.e., $g(t)>0$ for all $t \in \mathbb{R}$ and $\int_{-\infty}^{\infty} g(t)\,dt = 1$. From this, we construct a $s$-variate probability density

\begin{equation}
\label{prior1}
g(\boldsymbol{\theta}) = \prod_{j=1}^{s} g(\theta_{j}) \hbox{ for }\boldsymbol{\theta} = (\theta_{1},\ldots,\theta_{s}) \in \mathbb{R}^{s}.
\end{equation}

If the function$f$ defined on $\mathbb{R}^{s}$ is integrable with respect to $g$, we write $$I(f) =  \int_{\mathbb{R}^{s}}f(\boldsymbol{\theta}) g(\boldsymbol{\theta}) \, d\mathbf{\theta}.$$  The decomposition of an $s$ dimensional integrand $f(\cdot)$  is given by $$ f(\boldsymbol{\theta}) = \sum_{I\subseteq \{1,\ldots,s\}} f_{I}(\boldsymbol{\theta}), $$ where for nonempty subsets we have $$f_{I}(\boldsymbol{\theta}) = \int_{\mathbb{R}^{s-d}} f(\boldsymbol{\theta}) g(\boldsymbol{\theta}_{-I})\,d\mathbf{\theta}_{-I} -\sum_{J\subset I} f_{J}(\boldsymbol{\theta}),$$ where $d=|I|$. The ANOVA component $f_{\emptyset}(u)$ is simply the integral $I(f)$. The variance of each component is $$\sigma_{I}^{2} =  \int_{\mathbb{R}^{s}}f_{I}^{2}(\boldsymbol{\theta}) g(\boldsymbol{\theta}) \,d\mathbf{\theta}.$$

\subsection*{A3.2. anchored f-ANOVA (over unweighted spaces)}

Indeed, the drawbacks of standard and weighted ANOVA consists in the need to compute complex  high dimensional integrals. Alternatively, anchored ANOVA decomposition gives a computationally efficient way for the numerical evaluation of component functions in ANOVA (Griebel et.\ al.\ 2013, Tang et.\ al.\ 2014,  Yang et.\ al.\ (2012), and Gao and Hesthavan (2010)). The Dirac measure is used instead of Lebesgue measure and therefore the total weight is concentrated at a single point $\boldsymbol{c}$ , called the \emph{anchor point}. That is, the components which are to be integrated out are instead evaluated at the anchor point. Anchored f-ANOVA will yield easy approximation to the functional decompositions.

Let $\boldsymbol{c} = (c_1,\ldots,c_s)$. Then $f_{\emptyset}(\boldsymbol{\theta})$ is approximated as  $$f_{\emptyset}(\theta_1,\ldots,\theta_s) = f_{\emptyset}(c_1,\ldots,c_s),$$ and the function corresponding to the first component is $$f_{\theta_1}( \theta_1,\ldots,\theta_s) = f(\theta_1,c_2,\ldots,c_s),$$ and in general for any subset $I$, $$f_{I}(\boldsymbol{\theta}) = f(\boldsymbol{c}_{-I};\boldsymbol{\theta}_I)  -\sum_{J\subset I} f_{J}(\boldsymbol{\theta}),$$ where $f(\boldsymbol{c}_{-I};\boldsymbol{\theta}_I)$ represents the the value of $f(\boldsymbol{\theta})$ evaluated at anchor point $\boldsymbol{c}$ except for the variables involved in $I$.

 The variance of each component can then be approximated as  $$\sigma_{I}^{2} =  \int_{\mathbb{R}^{s}}f_{I}^{2}(\boldsymbol{\theta})\,d\boldsymbol{\theta}.$$

It is important to note that this is only an approximation, since for the anchored ANOVA, the orthogonality property is not valid (Tang et.\ al.\ 2014) and also that the accuracy of this approximation depends on the choice of the anchor point (Zhang et.\ al.\ 2010). Note that computing the variance components could still yield divergent integrals  if the support was unbounded.

\section*{Acknowledgements}

Paul Brown's research has been funded by the University of Waikato's doctoral scholarship.


\end{document}